\documentclass[11pt]{article}
\usepackage[margin=1in]{geometry}

\usepackage{graphicx,epsf,psfrag}
\usepackage{amsmath,amssymb}
\usepackage{amsfonts}
\usepackage{mathrsfs}
\usepackage{etoolbox}
\usepackage{comment}
\usepackage{amsthm}

\usepackage{latexsym}

\newcommand{\beq}{\begin{equation}}
\newcommand{\eeq}{\end{equation}}
\newcommand{\be}{\begin{equation}}
\newcommand{\ee}{\end{equation}}
\newcommand{\eps}{\epsilon}

\newcommand{\bi}{\begin{itemize}}
\newcommand{\ei}{\end{itemize}}

\newcommand{\calA}{\mathcal{A}}

\newcommand{\calC}{\mathcal{C}}
\newcommand{\calD}{\mathcal{D}}
\newcommand{\calE}{\mathcal{E}}

\newcommand{\calM}{\mathcal{M}}
\newcommand{\calN}{\mathcal{N}}

\newcommand{\calX}{\mathcal{X}}

\newcommand{\bbE}{\mathbb{E}}

\newcommand{\bbP}{\mathbb{P}}

\DeclareMathAlphabet{\mathbsf}{OT1}{cmss}{bx}{n}
\DeclareMathAlphabet{\mathssf}{OT1}{cmss}{m}{sl}

\DeclareSymbolFont{bsfletters}{OT1}{cmss}{bx}{n}  
\DeclareSymbolFont{ssfletters}{OT1}{cmss}{m}{n}
\DeclareMathSymbol{\bsfGamma}{0}{bsfletters}{'000}
\DeclareMathSymbol{\ssfGamma}{0}{ssfletters}{'000}
\DeclareMathSymbol{\bsfDelta}{0}{bsfletters}{'001}
\DeclareMathSymbol{\ssfDelta}{0}{ssfletters}{'001}
\DeclareMathSymbol{\bsfTheta}{0}{bsfletters}{'002}
\DeclareMathSymbol{\ssfTheta}{0}{ssfletters}{'002}
\DeclareMathSymbol{\bsfLambda}{0}{bsfletters}{'003}
\DeclareMathSymbol{\ssfLambda}{0}{ssfletters}{'003}
\DeclareMathSymbol{\bsfXi}{0}{bsfletters}{'004}
\DeclareMathSymbol{\ssfXi}{0}{ssfletters}{'004}
\DeclareMathSymbol{\bsfPi}{0}{bsfletters}{'005}
\DeclareMathSymbol{\ssfPi}{0}{ssfletters}{'005}
\DeclareMathSymbol{\bsfSigma}{0}{bsfletters}{'006}
\DeclareMathSymbol{\ssfSigma}{0}{ssfletters}{'006}
\DeclareMathSymbol{\bsfUpsilon}{0}{bsfletters}{'007}
\DeclareMathSymbol{\ssfUpsilon}{0}{ssfletters}{'007}
\DeclareMathSymbol{\bsfPhi}{0}{bsfletters}{'010}
\DeclareMathSymbol{\ssfPhi}{0}{ssfletters}{'010}
\DeclareMathSymbol{\bsfPsi}{0}{bsfletters}{'011}
\DeclareMathSymbol{\ssfPsi}{0}{ssfletters}{'011}
\DeclareMathSymbol{\bsfOmega}{0}{bsfletters}{'012}
\DeclareMathSymbol{\ssfOmega}{0}{ssfletters}{'012}

\newcommand{\tilL}{\tilde{L}}

\newcommand{\barf}{\bar{f}}

\newcommand{\veps}{\varepsilon}

\DeclareMathOperator{\var}{Var}

\ifcsmacro{theorem}{}{
\newtheorem{theorem}{Theorem}
\newtheorem{lemma}[theorem]{Lemma}
\newtheorem{proposition}[theorem]{Proposition}

\newtheorem{definition}{Definition} 
\newtheorem{assumption}{Assumption}

\newtheorem{remark}{Remark}

}

\newcommand{\qednew}{\nobreak \ifvmode \relax \else
      \ifdim\lastskip<1.5em \hskip-\lastskip
      \hskip1.5em plus0em minus0.5em \fi \nobreak
      \vrule height0.75em width0.5em depth0.25em\fi}

\usepackage{comment}
\usepackage[usenames,dvipsnames,svgnames]{xcolor}
\usepackage{microtype}
\usepackage{graphicx}
\usepackage{subfigure}
\usepackage{hyperref}       
\usepackage{booktabs} 
\usepackage{amsfonts}       
\usepackage{nicefrac}       
\usepackage{caption}
\usepackage{algorithm}
\usepackage{algpseudocode}
\usepackage{mathtools}
\usepackage{lipsum}
\usepackage{multirow}
\usepackage{bbding}

\usepackage{enumitem}

\newcommand{\KL}{\textup{KL}}
\newcommand{\V}{\textup{V}}

\newcommand{\uP}{\textup{P}}

\newcommand{\E}{\mathsf{E}}
\newcommand{\sE}{\mathsf{E}}
\newcommand{\BN}{\mathbb{N}}
\newcommand{\BP}{\mathbb{P}}
\newcommand{\BE}{\mathbb{E}}
\newcommand{\BR}{\mathbb{R}}

\newcommand{\dCL}[1]{\delta_{#1, \, \textup{SP-CLT}}}
\newcommand{\dSDa}[1]{\delta_{#1, \, \textup{SP-MSD}}^{(0)}}
\newcommand{\dSDb}[1]{\delta_{#1, \, \textup{SP-MSD}}^{(1)}}
\newcommand{\sdp}{saddle-point}
\newcommand{\Sdp}{Saddle-point}

\renewcommand{\eps}{\varepsilon}
\renewcommand{\epsilon}{\varepsilon}

\newcommand{\errt}{\mathrm{err}_{\mathrm{SP}}}
\DeclareMathOperator{\esup}{ess\,sup}

\newcommand{\define}{\triangleq}

\newcommand{\error}{\textup{error}}

\hypersetup{
    colorlinks,
    linkcolor={blue},
    citecolor={green!60!black},
    urlcolor={blue}
}

\usepackage[numbers]{natbib}

\usepackage[utf8]{inputenc}
\usepackage[T1]{fontenc}   
\usepackage{hyperref}
\usepackage{url}
\usepackage{booktabs}
\usepackage{amsfonts}
\usepackage{nicefrac}
\usepackage{microtype}
\usepackage{bbm}

\title{The Saddle-Point Accountant for Differential Privacy }

\author{%
Wael Alghamdi\thanks{Corresponding author, remaining authors in alphabetical order.}~\thanks{W.\hspace{-0.5pt} Alghamdi, F.\hspace{-0.5pt} P.\hspace{-0.5pt} Calmon, and J. F. Gomez are with the School of Engineering and Applied Science, Harvard University (emails: {alghamdi@g.harvard.edu, flavio@seas.harvard.edu, juangomez@g.harvard.edu})}~,
Shahab Asoodeh\thanks{S.\hspace{-0.5pt} Asoodeh is with the Department of Computing and Software,\hspace{-1pt} McMaster\hspace{-0.5pt} University\hspace{-1.5pt} (email:\hspace{-3pt} {asoodehs@mcmaster.ca})}~,
Flavio P. Calmon${}^\dagger$, Juan Felipe Gomez${}^\dagger$\\
Oliver Kosut\thanks{O.\hspace{-0.5pt} Kosut, L.\hspace{-0.5pt} Sankar, and F.\hspace{-0.5pt} Wei are with the School of Electrical, Computer, and Energy Engineering, Arizona State University (emails: {\{\hbox{okosut},lsankar,fwei16\}@asu.edu})}~,
Lalitha Sankar${}^\S$, and
Fei Wei${}^\S$
}

\date{}

\begin{document}

\maketitle

\begin{abstract}
    We introduce a new  differential privacy (DP) accountant called \emph{the saddle-point accountant} (SPA).  SPA approximates  privacy guarantees for the composition of DP mechanisms in an accurate and fast manner. Our approach is inspired by the saddle-point method---a ubiquitous numerical  technique in statistics. We prove rigorous performance guarantees by deriving upper and lower bounds for the approximation error offered by SPA. The crux of SPA is a combination of large-deviation methods with central limit theorems, which we derive via exponentially tilting the privacy loss random variables corresponding to the DP mechanisms. One key advantage of SPA is that it runs in constant time for the $n$-fold composition of a privacy mechanism.  Numerical experiments demonstrate that SPA achieves comparable accuracy to state-of-the-art accounting methods with a faster runtime.
\end{abstract}

\section{Introduction}

Differential privacy (DP) is a widely adopted standard for privacy-preserving machine learning (ML). Differentially private mechanisms used in ML tasks typically operate in the \emph{large-composition regime}, where  mechanisms are sequentially applied  many times to sensitive data.   For example, when training neural networks using stochastic gradient descent, DP can be ensured by clipping and adding Gaussian noise to each gradient update \citep{Abadi_MomentAccountant}. Here, a DP mechanism (gradient clipping plus noise) is applied hundreds of times to training data---even for simple datasets such as MNIST~\cite{LeCun} and CIFAR-10~\cite{CIFAR10}.  

Quantifying the privacy loss after a large number of compositions of DP mechanisms is a central challenge in privacy-preserving ML. This challenge has motivated the introduction of several new \emph{privacy accounting} techniques that aim to measure how privacy degrades with the number of queries computed over sensitive data (e.g., gradient updates). Examples of privacy accounting techniques include the moments accountant \citep{Abadi_MomentAccountant,RenyiDP}, convolution and FFT-based methods~\cite{koskela2020computing,koskela2021computing,koskela21a_FFT,Numerical_compositionDP, FasterPrivacyAccountant}, linearization of the underlying divergence measure~\cite{ConnectingDots}, characteristic function~\cite{zhu2022optimal},  and advanced composition formulas~\cite{dwork2010boosting,Dwork_Roth_book} (see Section~\ref{sec:setup} for a more detailed discussion on the related work). Most privacy accountants provide a bound on the $\eps$ and $\delta$ parameters of a variant of DP called \emph{approximate DP}, defined next. Throughout, we let $\calM(D)$ denote the output random variable of a randomized mechanism $\calM$ running on the dataset $D$.

\begin{definition}[$(\veps,\delta)$-DP~\cite{Dwork_Calibration,Dwork-OurData}] \label{def:DP}A mechanism (i.e., randomized algorithm) $\calM$ is \emph{$(\varepsilon,\delta)$-differentially private} (DP) if, for every pair of neighboring datasets $D$ and $D'$ and event $E$, we have $\BP\left[ \calM(D)\in E \right] \le e^{\varepsilon} \ \BP\left[ \calM(D') \in E \right] + \delta$.
\end{definition}

Any privacy mechanism $\calM$ induces a collection of $(\eps, \delta)$-DP guarantees; in other words, for each value of $\eps\geq 0$, there exists the best (i.e., smallest) achievable $\delta$ such that the mechanism is $(\eps, \delta)$-DP. Since such a $\delta$ depends on $\eps$, we write it as $\delta_\calM(\eps)$ and call the function $\delta_{\calM}: \BR_+ \to [0,1]$ the \textit{optimal privacy curve} of $\calM$. A naive approach for characterizing the privacy curve  $\delta_\calM(\eps)$ for a given mechanism is to invoke Definition~\ref{def:DP} and directly compute, for a fixed $\eps$, the corresponding $\delta$ for each pair of neighboring datasets $D$ and $D'$ then maximizing over all such values of $\delta$ as $(D,D')$ vary.

The  \textit{adaptive} composition of two mechanisms $\calM_1$ and $\calM_2$ is given by the mechanism 
\begin{equation}
\calM_1\circ \calM_2(D) \define
(\calM_1(D), \calM_2(D, \calM_1(D))),
\end{equation} 
that is, $\calM_2$ can look at both the dataset and the output of $\calM_1$. Let $\calM^{\circ n}$ denote the adaptive composition of the same mechanism $\calM$ for $n$ times. Characterizing tight bounds for  the privacy curve $\delta_{\calM^{\circ n}}(\epsilon)$ of $\calM^{\circ n}$ is a central question in DP literature~\cite{Dwork-OurData, dwork2010boosting, Kairouz_Composition_TIT, kairouz_Composition, sommer2019privacy, koskela2020computing,koskela2021computing,koskela21a_FFT,Numerical_compositionDP, FasterPrivacyAccountant, zhu2022optimal, ConnectingDots}.

We develop a framework, called \emph{{\sdp} accountant} (SPA), for accurately and efficiently approximating $\delta_{\calM^{\circ n}}$ using the \textit{{\sdp}} technique---a well-known technique used in statistics to approximate tail probabilities of complicated distributions, see, e.g., \cite{reid_saddlepoint_1988,daniels_tail_1987,small_Sample} for an overview. When approximating statistics involving estimates of tail probabilities, the {\sdp} method is more precise than the central limit theorem (CLT) and Edgeworth expansion-based approaches~\citep{kolassa_series_2006}, and is one of the go-to numerical techniques in applied statistics---yet, until now, it has not been adopted in DP. Informally speaking, SPA combines the large deviation behavior of the moments accountant \cite{Abadi_MomentAccountant} with the CLT approach of \cite{sommer2019privacy,GaussianDP}, to derive an accountant that dramatically outperforms both, with only marginal more computational cost.

The SPA uses the effective concept of a \emph{privacy loss random variable} (PLRV) (see, e.g.,~\cite{koskela2020computing,koskela2021computing,koskela21a_FFT,Numerical_compositionDP, FasterPrivacyAccountant, zhu2022optimal}). To this end, one begins by first identifying a dominating pair of distributions $(P, Q)$ for the mechanism $\calM$:
\begin{equation} \label{eq:P,Q}
    \BP\left[ \calM(D)\in E \right] - e^{\varepsilon} \ \BP\left[ \calM(D') \in E \right] \leq P(E) - e^\eps Q(E)
\end{equation} for any event $E$ and pair of neighboring datasets $D$ and $D'$ (see Definition~\ref{def:dominating} for more details). Such (not necessarily unique) pair exists for all mechanisms  \cite{zhu2022optimal} and provides an alternative way for computing  $\delta_{\calM}(\eps)$ in terms of a PLRV. More precisely, it can be shown that $\calM$ is $(\eps, \delta_L(\eps))$-DP with \begin{equation} \label{eq:delta L}
\delta_L(\eps) \define \mathbb E\left[ \left(1-e^{-(L-\eps)} \right)^+\right],
\end{equation}
where $a^+ \define \max\{0, a\}$ and
\begin{equation}
L \define \log\frac{\text{d} P}{\text{d} Q}(X),
\end{equation}
with $X\sim P$. Here, $L$ is called a PLRV, and we call $\delta_L$ as given by~\eqref{eq:delta L} its \emph{induced privacy curve}. That $\calM$ is $(\eps, \delta_L(\eps))$-DP can be written mathematically as $\delta_{\calM}\le \delta_L$. Further, if $(P,Q)$ is tightly dominating in~\eqref{eq:P,Q} (see Definition~\ref{def:dominating}), then the induced privacy curve of $L$ coincides with the optimal privacy curve for $\calM$, i.e., we have $\delta_\calM = \delta_L$.

Privacy accounting is particularly important in the composition setting. In this setting, we can form a PLRV for the composed mechanism that splits additively. In other words, $L^{(n)} \define L_1+\cdots+L_n$, where $L_1,\cdots,L_n$ are i.i.d., is a PLRV for the $n$-fold composition $\calM^{\circ n}$. The induced privacy curve $\delta_{L^{(n)}}$ gives then a privacy guarantee for $\calM^{\circ n}$, i.e., $\delta_{\calM^{\circ n}}\le \delta_{L^{(n)}}$. Furthermore, this inequality is tight in general~\cite{GaussianDP}: there exists a mechanism $\calM$ for which the optimal privacy curve of the composition $\calM^{\circ n}$ is equal to $\delta_{L^{(n)}}$, i.e., $\delta_{\calM^{\circ n}} = \delta_{L^{(n)}}$. For these reasons, we focus on approximating the curve $\delta_{L^{(n)}}$, which we call the \emph{composition curve} (see Definition~\ref{Def:CompositionCurve}). It is important to note that all existing accounting methods provide bounds on the composition curve, rather than the optimal privacy curve $\delta_{\calM^{\circ n}}$ itself.

To justify the use of the {\sdp} method in approximating the privacy curve $\delta_{L^{(n)}}$, we show that $\delta_{L^{(n)}}$ can be expressed in terms of an integral on the complex plane by means of the Fourier transform. We then provide two different approximations for this integral that  involve the cumulant generating function (CGF)  
\begin{equation}
    K_L(t) \define \log \mathbb E\left[e^{tL}\right]
\end{equation}
of the PLRV corresponding to the mechanism $\calM$. Both approximations require a judicious choice of the integration path in the complex plane: integration is done over the line parallel to the imaginary axis with real part corresponding to the \emph{{\sdp}} of the integrand. The {\sdp}  is given by the unique solution of 
\begin{equation}\label{Saddle_point}
    \boxed{n K'_L(t) = \eps + \frac{1}{t} + \frac{1}{t+1}.}
\end{equation}

The first approximation is based on expanding the integrand around the {\sdp}  via Taylor expansion (see Section~\ref{sec:saddlepoint}). In Mathematical Physics, this corresponds to the well-known \emph{method of steepest descent} (MSD)~\cite{mathematical_physics_jeffrey}. This approach leads to  the following approximation: 
\begin{equation}\label{SPA_approx1}
    \boxed{\delta_{L^{(n)}}(\eps) \approx \dSDa{L^{(n)}}(\eps) \define  \frac{e^{nK_L(t) - \eps t}}{\sqrt{2\pi [nt^2(1+t)^2 K''_L(t) + t^2 + (1+t)^2]}}, }
\end{equation}
where $t$ is the unique solution of \eqref{Saddle_point}. We also derive a more accurate approximation based on the same approach that includes additional terms, denoted $\delta^{(1)}_{L^{(n)},\text{SP-MSD}}(\eps)$ (see Sec.~\ref{sec:steepest_descent} for details).

The second approximation involves recentering the density of $L$ by introducing its exponentially tilted version $\tilde L$ and then applying CLT to $\tilde L$ (see Section~\ref{sec:CLT}). As discussed in Section~\ref{sec:exp_tilting}, tilting the distribution is equivalent to selecting the integration path in the complex plane. By choosing the tilting parameter $t$ that satisfies \eqref{Saddle_point}, this approach leads to the following approximation:
\begin{equation}\label{SPA_approx2}
    \boxed{\delta_{L^{(n)}}(\eps) \approx \dCL{L^{(n)}}(\eps)
\define e^{nK_L(t)-\eps t} ~\bbE\Big[ e^{-t(Z-\varepsilon)}\big(1-e^{-(Z-\varepsilon)}\big)^+\Big],}
\end{equation}
where $Z$ is a Gaussian random variable with mean $nK'_L(t)$ and variance $nK''_L(t)$. Notice that the expectation in \eqref{SPA_approx2} can be analytically computed. 
While the first approximation can exhibit a better accuracy in the numerical experiments in Section~\ref{sec:numerical}, the second one is more amenable toward the error analysis given in Section~\ref{sec:BE}. 
Both approximations hold even if we consider composition of $n$ different mechanisms $\calM_1, \dots, \calM_n$, except that in this case we replace $nK_L(t)$ with $\sum_{i=1}^n K_{L_i}(t)$, where $L_i$ is a PLRV corresponding to $\calM_i$.

The proposed framework has the following main features: 
\begin{enumerate}
\item \textbf{Precise DP accounting:}
Both $\dSDa{L^{(n)}}(\eps)$ and $\dCL{L^{(n)}}(\eps)$ provide similar precision to the state-of-the-art accountants. To illustrate this, in Figure~\ref{fig:DP-SGD-exp} we compare both approximations with the moments accountant~\cite{Abadi_MomentAccountant}, CLT-based accountant for Gaussian DP (GDP)~\cite{GaussianDP}, and the accountant proposed by Gopi et al.~\cite{Numerical_compositionDP} that is known to be the state-of-the-art. This figure illustrates that our framework significantly outperforms both the moments accountant and GDP (which underreports the privacy curve); it is also comparable with the state-of-the-art accountant of~\cite{Numerical_compositionDP}. In addition, a remarkable advantage of our two approximations over the state-of-the-art is that it provides a precise estimate for $\delta_{L^{(n)}}(\eps)$ consistently for all values of parameters, whereas the accountant by Gopi et al.~\cite{Numerical_compositionDP} suffers from floating point inaccuracies when the true value of $\delta$ is small (see Figure~\ref{fig:dve}).\footnote{See Appendix B in \cite{Numerical_compositionDP} and their GitHub \cite{PRV_break} for further discussion on the limitations of the implementation.}

\item \textbf{Lower runtime complexity compared to FFT-based accountants:} The computational complexity for computing both $\dSDa{L^{(n)}}(\eps)$ and $\dCL{L^{(n)}}(\eps)$ is \textit{independent of the number of compositions} for the $n$-fold composition setting. This is in stark contrast with the state-of-the-art~\cite{Numerical_compositionDP} whose complexity is\footnote{We use standard asymptotic notation: $f=O(g)$ means $\limsup_{n\to \infty} |f(n)|/g(n) < \infty$, $f=o(g)$ or $f\ll g$ means $\lim_{n\to \infty} f(n)/g(n) = 0$, and $f \sim g$ means $\lim_{n\to \infty} f(n)/g(n) = 1$.} $O(\sqrt{n} \log n)$ and a recent followup work which achieved $O(\text{polylog} \: n)$~\cite{FasterPrivacyAccountant}.

\item \textbf{The best of both worlds---combining large deviation and central limit approaches:} There are two main approaches to approximating or bounding expectations of sums of independent random variables. The first is the large deviation approach, which uses the moment generating function to approximate the probability of very unlikely events. The second is the central limit theorem (CLT), which approximates a random variable by a Gaussian with the same mean and variance. For DP accounting, the large deviation approach led to the moments accountant \cite{Abadi_MomentAccountant}; the CLT approach led to Gaussian DP \cite{sommer2019privacy,GaussianDP}. Both these accountant methods can be computed in constant time, but their accuracy is far less than the state-of-the-art FFT accountants. The {\sdp} method can be viewed as a combination of two basic approaches: maintaining from large deviations the ability to handle very small values of $\delta$, as well as the precise guarantees of the CLT. The resulting accountant achieves better accuracy than either approach on its own, while maintaining constant runtime.
\end{enumerate}

\begin{figure}[ht]
    \centering
    \begin{minipage}{0.5\textwidth}
        \centering
        \includegraphics[width=1\textwidth]{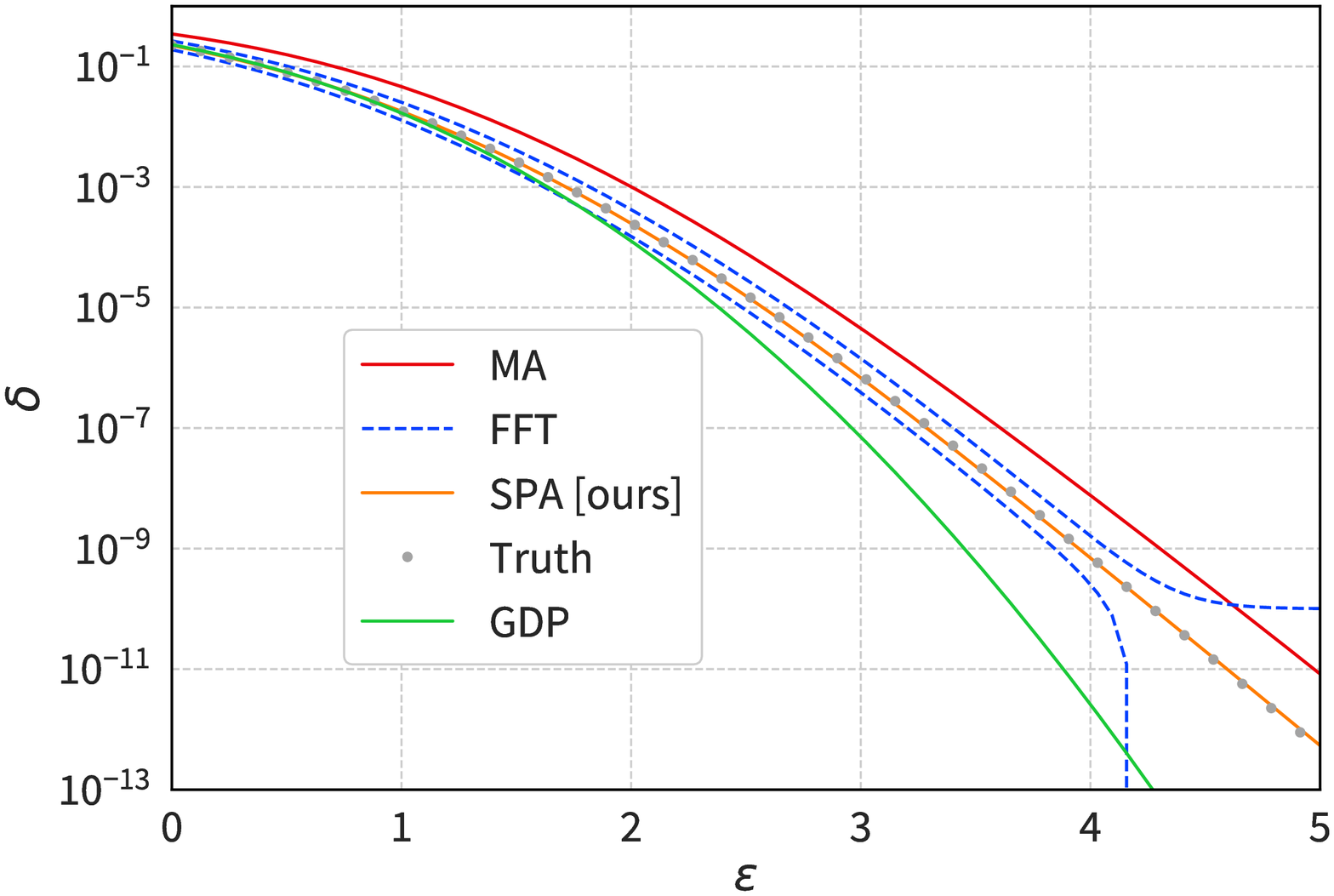}
        \captionsetup{width=0.9\textwidth}
        \caption*{(a) 
        }
    \end{minipage}\hfill
    \begin{minipage}{0.5\textwidth}
        \centering
        \includegraphics[width=1\textwidth]{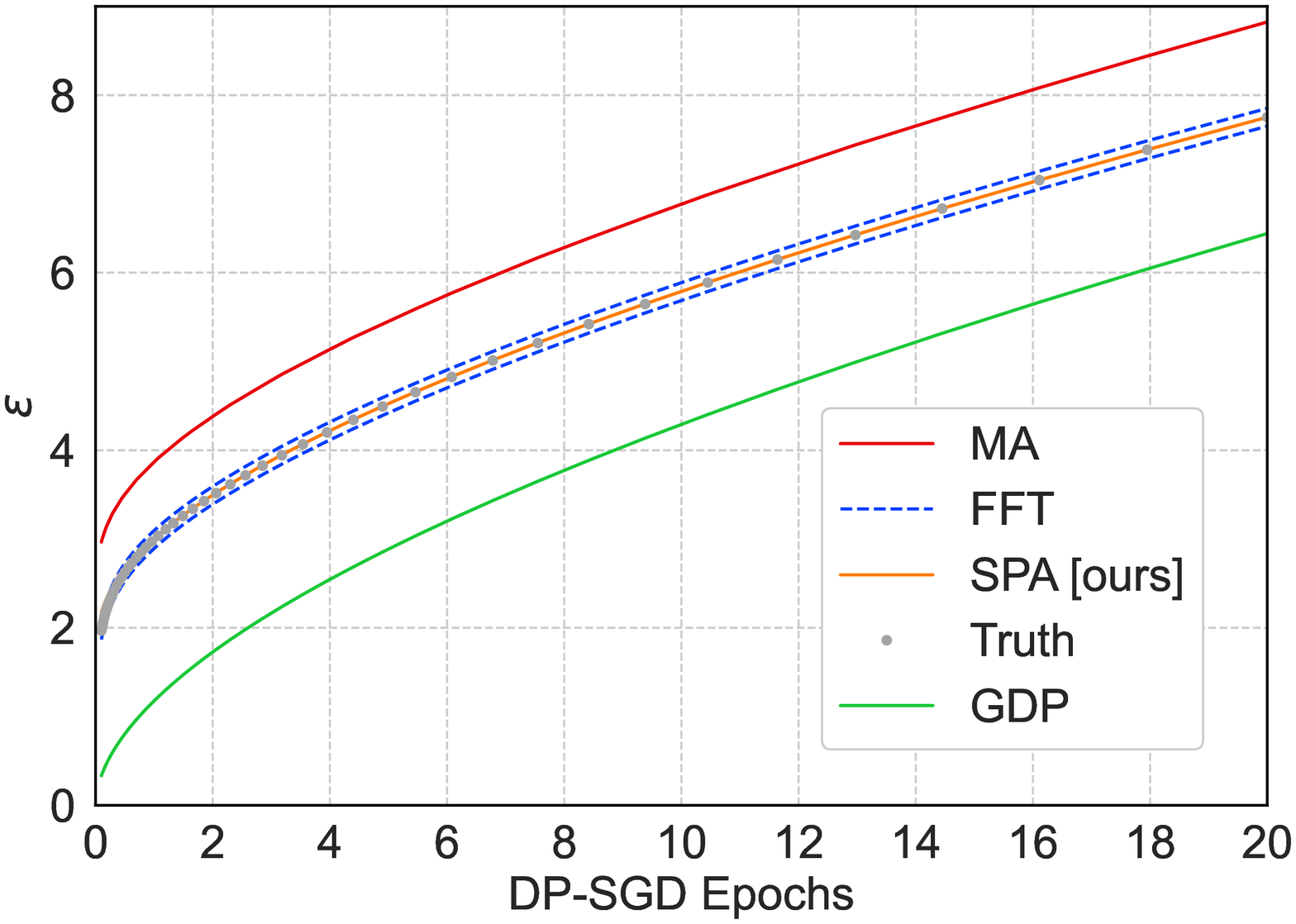} 
        \captionsetup{width=0.9\textwidth}
        \caption*{(b)
        }
    \end{minipage}
    \caption{Pane (a) shows the subsampled Gaussian privacy curve after 2000 compositions with noise $\sigma = 1$, and subsampling rate $\lambda = 10^{-2}$. We plot three approximations to $\delta_{L^{(n)}}$ given by \eqref{SPA_approx1}, \eqref{SPA_approx2}, and $\delta^{(1)}_{L^{(n)},\text{SP-MSD}}(\eps)$ (defined in  Section~\ref{sec:steepest_descent}). The three approximations are labelled ``SPA [ours]'' since they are indistinguishable.   Pane~(b) shows a simulation of DP-SGD with noise $\sigma = 0.65$, subsampling rate $\lambda = 10^{-2}$, and privacy parameter $\delta = 10^{-5}$. We again plot all three approximations to  $\delta_{L^{(n)}}$ and label them all ``SPA [ours].'' The ``Truth'' in both plots is computed from~\eqref{eq:fourier-tilt} using high precision Gaussian quadrature integration (see Section~\ref{sec:numerical} for details). ``MA'' refers to the moments accountant~\cite{Abadi_MomentAccountant} (with the latest implementation, see the ``Baselines'' paragraph in Section~\ref{sec:numerical}); 
    ``GDP'' to Gaussian DP \cite{GaussianDP}; 
    and FFT to the accountant proposed in~\cite{Numerical_compositionDP} with $\epsilon_{\error} = 0.1$ and $\delta_{\error} = 10^{-10}$.} 
    \label{fig:DP-SGD-exp}
\end{figure}

These features of the proposed SPA will be made mathematically precise in the sequel. The rest of the paper is organized as follows. First, we review central DP notions in Section~\ref{sec:setup}. Then, we derive the SPA in Section~\ref{sec:exp_tilting}, as well as a more precise variation of~\eqref{SPA_approx1} that accounts for higher-order derivatives of the CGF $K(t)$. Error bounds are given in Section~\ref{sec:BE}, where they are also theoretically compared to a standard CLT approach. Finally, we perform numerical experiments comparing SPA with state-of-the-art DP accounting techniques in Section~\ref{sec:numerical}.

\section{Setup and Related Work} \label{sec:setup}

We start with formally defining several central objects of DP accounting: privacy curves, composition, and privacy loss random variables. To do so more compactly, we first give the definition of the \textit{hockey-stick divergence}.  
\begin{definition}[Hockey-stick divergence] \label{Hockey-stick divergence}
Given a constant $\gamma\geq 1$ and a pair of probability measures $(P, Q)$ defined on $(\calX,\Sigma)$, the hockey-stick divergence from $P$ to $Q$ is defined as 
\begin{equation}
    \E_\gamma(P\|Q) \define (P-\gamma Q)^+(\calX) = \int_{\calX} \ \textnormal{d}(P - \gamma Q)^+
=\sup_{\calA \in \Sigma} \ P(A)-\gamma\, Q(A),
\end{equation}
where $\mu^+$ is the positive part of a signed measure $\mu$.
\end{definition}
This divergence---also known as the $\E_\gamma$-divergence in the information theory literature, see, e.g.,~\cite{polyanskiy2010channel, E_gamma}---provides an alternative expression for the approximate differential privacy, defined in Definition~\ref{def:DP}.

\begin{lemma}(\cite{Barthe:2013_Beyond_DP}) \label{lem:DP Egamma}
A randomized mechanism $\calM$ is $(\veps,\delta)$-DP, for $\veps\geq 0$ and $\delta\in [0,1]$, if and only if \footnote{We are abusing notation slightly by using $\calM(D)$ and $\calM(D')$ to denote the distributions of the output of the mechanisms running on neighboring datasets $D$ and $D'$, respectively. Also, the notation $D\sim D'$ is used here to indicate that the datasets $D$ and $D'$ are neighbors.}
\begin{equation}
    \sup_{D\sim D'} \ \E_{e^\veps}(\calM(D) \| \calM(D')) \le \delta.
\end{equation}
\end{lemma}

\paragraph{Privacy curves.} Given a mechanism $\calM$, the goal is to find a collection of pairs $(\veps,\delta)$ for which $\calM$ is $(\veps,\delta)$-DP, which gives rise to the concept of a privacy curve.

\begin{definition}[Optimal privacy curve] \label{def:OptimalPC}
The \emph{optimal privacy curve} of a mechanism $\calM$ is the function $\delta_{\calM} \colon [0,\infty)\to [0,1]$ such that, for each $\varepsilon \ge 0$, $\delta=\delta_{\calM}(\varepsilon)$ is the smallest number for which $\calM$ is $(\varepsilon,\delta)$-DP.
In other words,
\begin{equation}
    \delta_{\calM}(\varepsilon) \define \sup_{D\sim D'} \ \E_{e^\eps}(\calM(D)\|\calM(D')).
\end{equation}
We also define the (left-)inverse curve $\varepsilon_{\calM}(\delta)$ by
\begin{equation}
    \varepsilon_{\calM}(\delta) \define \inf\left\{ \,  \varepsilon \ge 0 \  : \ \delta \le  \delta_{\calM}(\varepsilon) \, \right\}.
\end{equation}
\end{definition}

According to this definition, $\calM$ is $(\veps,\delta)$-DP for \emph{every} $\veps\ge 0$ and $\delta \in [\delta_\calM(\veps),1]$. An important example is the Gaussian mechanism $\calM(D) = \calN(f(D),\sigma^2I_d)$, where $f$ is an $\BR^d$-valued query function having global $\ell_2$-sensitivity $1$. The optimal privacy curve of this mechanism is given by 
$\delta_{\calM}(\veps) = \Phi(-\veps \sigma + \frac{1}{2\sigma}) - e^\veps \ \Phi(-\veps \sigma - \frac{1}{2\sigma})$ ~\cite[Theorem~8]{Improving_Gaussian}, where $\Phi(x) = \frac{1}{\sqrt{2\pi}}\int_{-\infty}^xe^{-t^2/2}\text{d}t$.
According to Definition~\ref{def:OptimalPC}, one needs to compute $\E_{e^\eps}(\calM(D)\|\calM(D'))$ for all pairs of neighboring datasets $D$ and $D'$ to characterize $\delta_\calM(\eps)$. An alternative approach is to identify a pair of distributions $(P, Q)$ that dominates $\E_{e^\eps}(\calM(D)\|\calM(D'))$ for any neighboring $D$ and $D'$. This leads to the following definition. 

\begin{definition}[Dominating pair of distributions \cite{zhu2022optimal}]\label{def:dominating}
A pair of probability measures $(P, Q)$ is said to be a pair of dominating distributions for $\calM$ if 
for every $\eps\geq 0$ 
\begin{equation}
    \sup_{D\sim D'}~\E_{e^\eps}(\calM(D) \| \calM(D')) \le \E_{e^\eps}(P\| Q). 
\end{equation}
 If equality is achieved for every $\eps\geq 0$, then $(P,Q)$ is said to be \emph{tightly} dominating distributions for $\calM$.
\end{definition}

It is known~\cite[Proposition~10]{zhu2022optimal} that all mechanisms have a pair of tightly dominating distributions. This therefore leads to an alternative characterization of the optimal privacy curve, namely,  $\delta_\calM(\eps) = \sE_{e^\eps}(P\|Q)$ for tightly dominating distributions $(P, Q)$. It is not hard to show that 
\begin{equation}\label{HS_diver_L}
    \sE_{e^{\eps}}(P\|Q)
    =\bbE\left[\left(1-e^{\eps} \frac{\text{d}Q}{\text{d}P}(X)\right)^+\right]
    =\bbE\left[\left(1- e^{-(L-\eps)}\right)^+\right],
\end{equation}
where $L\define \log\frac{\text{d}P}{\text{d}Q}(X)$,
with $X\sim P$, is a PLRV for mechanism $\calM$, provided the equivalence of measures $P\sim Q$. 

We associate with each random variable $L$ its induced privacy curve, as follows.

\begin{definition}[Privacy curve of a random variable] \label{def:PLRV}
For a random variable $L$, its induced privacy curve $\delta_L:[0,\infty)\to [0,1]$ is defined by
\begin{equation}
    \delta_L(\eps) := \BE\left[\left(1- e^{-(L-\eps)}\right)^+\right],
\end{equation}
and we denote the (left-)inverse by $\varepsilon_{L}(\delta) := \inf\left\{ \varepsilon \ge 0 ~ : ~ \delta \le  \delta_{L}(\varepsilon) \right\}$.
\end{definition}

Thus, the optimal privacy curve of any mechanism can be characterized by its PLRVs: if $(P,Q)$ is a tightly dominating pair for $\calM$, then the optimal privacy curve of $\calM$ is given in terms of its PLRV $L=\log \frac{\text{d}P}{\text{d}Q}(X)$, $X\sim P$, by $\delta_\calM = \delta_L$.

\paragraph{Composition of DP mechanisms.} 
The \emph{adaptive composition} of two mechanisms $\calM_1$ and $\calM_2$ is given by the mechanism $(\calM_1\circ \calM_2)(D) := (\calM_1(D),\calM_2(D,\calM_1(D)))$,  that is, $\calM_2$ takes as input both $D$ and the output of $\calM_1$.  We denote by $\calM^{\circ n}$ the adaptive composition of a single mechanism $\calM$ for $n$ times. In contrast, the composition is called \emph{non-adaptive} if the output of $\calM_2$ depends on $\calM_1(D)$ only through $D$. 

It is not hard to show that $\delta_{\calM^{\circ n}}(n\veps) \le n \delta_{\calM}(\veps)$, a result that is commonly called basic composition~\cite[Theorem~1]{Dwork-OurData}. A tighter composition theorem, called advanced composition, was shown in~\cite{dwork2010boosting}. 
The exact characterization of the privacy parameters of $\calM^{\circ n}$ have been derived by Kairouz et al. \cite[Theorem~9]{Kairouz_Composition_TIT} in terms of the privacy parameters of $\calM$. It was shown by Murtagh and Vadhan \cite[Theorem~1.5]{Vadhan_Murtagh} that computing exact parameters for $\calM_1\circ \dots\circ \calM_n$ is in general \#P-complete, and thus, infeasible.  Nevertheless, the authors in~\cite{Vadhan_Murtagh} proposed a polynomial-time algorithm that is capable of approximating $\delta_{\calM_1\circ\dots\circ\calM_n}(\eps)$ for a given $\eps$ to arbitrary accuracy. However, its complexity scales at most as $n^3$, rendering it inefficient when considering composition of many mechanisms (i.e., large $n$).  
The lack of efficient algorithms for such an approximation has spurred several follow-up works, e.g.,~\cite{GaussianDP,koskela2020computing, koskela2021computing,koskela21a_FFT,Numerical_compositionDP, FasterPrivacyAccountant, ConnectingDots} to name a few. Most such works can best be described via the PLRV through the following theorem.     

\begin{theorem}[{\cite[Theorem~3.2]{GaussianDP}}]\label{thm:Dominating_Composition} Let $(P_j, Q_j)$ be a pair of tightly dominating distributions for mechanism $\calM_j$ for $j\in \{1, \dots, n\}$. Then $(P_1\times\dots\times P_n,  Q_1\times\dots\times Q_n)$ is a pair of dominating  distributions for $\calM_1\circ \dots\circ\calM_n$.  
\end{theorem}

While this theorem does not necessarily provide a pair of tightly dominating distributions for the adaptive composition of mechanisms, it still leads to an upper bound for its privacy curve. In particular, together with Definitions~\ref{def:OptimalPC} and \ref{def:dominating}, Theorem~\ref{thm:Dominating_Composition} implies that  \begin{equation}
    \delta_{\calM_1\circ\dots\circ\calM_n}(\eps)\leq \E_{e^{\eps}}\left( P_1\times \cdots \times P_n, Q_1 \times \cdots \times Q_n \right),
\end{equation}
where $(P_j, Q_j)$ is a pair of tightly dominating distributions for $\calM_j$. Invoking the alternative expression for hockey-stick divergence given in \eqref{HS_diver_L}, and Definition~\ref{def:PLRV}, we can hence write
\begin{equation}\label{eq:UB_composition}
    \delta_{\calM_1\circ\dots\circ\calM_n}(\eps)\leq \bbE\left[\left(1- e^{-(L^{(n)}-\eps)}\right)^+\right]=\delta_{L^{(n)}}(\eps),
\end{equation}
where 
\begin{equation}\label{eq:Ln}
    L^{(n)} \define \sum_{i=1}^n L_i,
\end{equation}
and $L_i \define \log \frac{\text{d}P_i}{\text{d}Q_i}(X_i)$, $X_i\sim P_i$, are independent PLRVs for the $\calM_i$. This upper bound is the foundation for most recent numerical composition results (such as \cite{sommer2019privacy, GaussianDP, koskela2020computing, Numerical_compositionDP,zhu2022optimal}) and similarly serves as the linchpin for our proposed composition algorithm outlined in the subsequent sections. As a result, we give the upper bound in \eqref{eq:UB_composition} a name.
\begin{definition}\label{Def:CompositionCurve}
Let $(P_j, Q_j)$ be a pair of tightly dominating distributions\footnote{Recall from \cite[Proposition~10]{zhu2022optimal} that all mechanisms have such a pair.} for mechanism $\calM_j$ for $j\in \{1, \dots, n\}$, and consider the adaptive composition $\calM^{(n)}=\calM_1\circ \cdots \circ \calM_n$. Then, we define the \emph{composition curve} $\bar{\delta}_{\calM^{(n)}}:[0, \infty)\to [0,1]$ as  
\begin{equation}
    \bar{\delta}_{\calM^{(n)}}(\eps) \define \sE_{e^\eps}(P_1\times\dots\times P_n\|  Q_1\times\dots\times Q_n) = \bbE\left[\left(1- e^{-(L^{(n)}-\eps)}\right)^+\right]=\delta_{L^{(n)}}(\eps),
\end{equation}
where $L^{(n)}$ is defined in \eqref{eq:Ln}.
\end{definition}

\paragraph{Subsampling.} 
Subsampling is a fundamental tool in the analysis of differentially private mechanisms.
Informally speaking, subsampling entails applying a differentially private mechanism to a
small set of randomly sampled datapoints from a given dataset. There are several ways of formally defining 
the subsampling operator, see, e.g., \cite{Balle:Subsampling}. The most well-known one, Poisson subsampling, is parameterized by the subsampling rate $\lambda\in (0,1]$ which indicates the probability of selecting a datapoint. More formally, the subsampled datapoints from a dataset $D$ can be expressed as $\{x\in D~:~B_x = 1\}$, where $B_x$ is a Bernoulli random variable with parameter $\lambda$ independent for each $x\in D$. Given any mechanism $\calM$, we define the subsampled mechanism $\calM_\lambda$ as the composition of $\calM$ and the Poisson subsampling operator. 
Characterizing the privacy guarantees of subsampled mechanisms is the subject of ``privacy amplification by subsampling'' principle \cite{Shiva_subsampling}. This principle is well-studied particularly for characterizing the privacy guarantees of subsampled Gaussian mechanisms in the context of a variant of differential privacy, namely, R\'enyi differential privacy~\cite{Poisson_Subsample_RDP, Abadi_MomentAccountant, mironov2019r}.  We can mirror their formulation to characterize $\eps$ and $\delta$ for the subsampled Gaussian mechanisms. Recall that a Gaussian mechanism satisfies $\calM(\calD)=\calN(f(D),\sigma^2 I_d)$ where $f$ is a query function with $\ell_2$-sensitivity $1$. For the \emph{subsampled} Gaussian, the optimal privacy curve (of a single composition) is
\begin{equation}\label{eq:subsampling1}
    \delta_{\calM_\lambda}(\eps) = \max\big\{\sE_{e^\eps}(P\|Q), \sE_{e^\eps}(Q\|P)\big\},
\end{equation}
where $P = \calN(0, \sigma^2 I_d)$ and $Q = (1-\lambda) P + \lambda P'$, and $P'\sim \calN(\boldsymbol{e}_1, \sigma^2 I_d)$ where $\boldsymbol{e}_1$ is the first standard basis vector. 
In the following lemma, we show that the above maximum is always attained by $\sE_{e^\eps}(Q\|P)$ for any $\eps\geq 0$, and that it holds for a larger family of DP mechanisms (including Gaussian and Laplace mechanisms). A similar ordering bound was proved in \cite[Theorem 5]{Subsampling_renyi} for the R\'enyi divergence.

\begin{lemma} \label{lem:subsampling}
Fix a Borel probability measure $P$ over $\BR^n$ that is symmetric around the origin (i.e., $P(\calA)=P(-\calA)$ for every Borel $\calA\subset \BR^n$), and fix constants $(s,\lambda,\gamma)\in \mathbb{R}^n\times [0,1]\times [1,\infty)$. Let $T_sP$ be the probability measure given by $(T_sP)(\calA)=P(\calA-s)$, and let $Q= (1-\lambda)P+\lambda T_s P$. We have the inequality
\begin{equation}
\mathsf{E}_\gamma(P\|Q) \le \mathsf{E}_{\gamma}(Q\|P)
\label{eq:subsample_bound}
\end{equation}
Further, equality holds if and only if $  (\gamma-1)\: \lambda \: \|s\| \: \E_\gamma(Q\|P)\: = 0$.
\end{lemma}
\begin{proof}
See Appendix~\ref{app:subsampling}.
\end{proof}

In light of this lemma, the privacy guarantee of a subsampled Gaussian mechanism is fully characterized by computing only $\sE_{e^\eps}((1-\lambda)P + \lambda T_s P\|P)$, where $P = \calN(0, \sigma^2 I_d)$. Based on this result, for our numerical results in Section~\ref{sec:numerical}, we only compute the saddle-point accountant with this order of $P$ and $Q$.

\section{{\Sdp} Accountant Methodology}\label{sec:exp_tilting}

This section describes the foundations of our {\sdp}  accountants. We first introduce our key tool, namely, exponential tilting. As a warm-up, we use exponential tilting to rederive---in a new way---two elementary variations of the moments accountant~\cite{Abadi_MomentAccountant} that render tighter privacy guarantees. We then discuss the connection between exponential tilting and the method of steepest descent, and describe in details how these two methods lead to our {\sdp} accountants.

We assume that, for the purposes of bounding or approximating the optimal privacy curve $\delta_{\calM}(\epsilon)$ for a mechanism $\calM$, we have access to a PLRV $L$ such that $\delta_{\calM}(\epsilon)\le \delta_L(\epsilon)$. In most cases, the relevant variable is $L^{(n)}=L_1+\cdots+L_n$, such that, as discussed above, $\delta_{L^{(n)}}$ is the composition curve. However, in this section we derive approximations for any variable $L$. We assume that the distribution of $L$ is known to an extent that expectations of functions of $L$---i.e., $\bbE[f(L)]$---can be computed. We proceed to derive approximations and bounds on $\delta_L(\epsilon)$ based on such expectations.

\subsection{Exponential Tilting}\label{sec:tilting}

The moment generating function (MGF) for a PLRV $L$ is given by
\begin{equation}
    M_L(z)\triangleq \bbE[e^{zL}]
\end{equation}
and the cumulant generating function (CGF) is given by
\begin{equation}
    K_L(z)\triangleq \log M_L(z).
\end{equation}
Let $\calC$ be the set of positive real numbers $t$ where these functions are defined and finite for $L$. Note that $\calC$ is an interval,\footnote{This can be seen by conditioning on the sign of $L$ to obtain $M_L(t')\le M_L(t)+1$ for $0\le t' \le t$.} and we assume that this is a nontrivial interval (i.e., not just $\{0\}$). 
We may also evaluate $M_L$ at complex arguments  for any $z$  such that $\Re(z)\in \mathcal{C}$,  since for $z=t+is$ and $t\in \mathcal{C}$ we have
\begin{equation}
    \left| M_L(z) \right|  \leq M_L(t).
\end{equation}
Of course, $M_L(is)$ is the characteristic function of $L$.

We may rewrite the induced privacy curve for $L$ using the Plancherel-Parseval identity as\footnote{Technically, this is an improper use of the identity, as the function $(1-e^{-x})^+$ is not integrable. We consider \eqref{eq:parseval} to be a formal equality, and we will subsequently provide a justification that this equality does indeed hold, provided the integral in \eqref{eq:parseval} is interpreted as taking a contour through the complex plane that avoids the pole at $0$.} 
 \begin{align}
     \delta_{L}(\epsilon)&= \mathbb{E} \left[\left(1-e^{-(L-\epsilon)}\right)^+\right]\\
     &= \frac{1}{2\pi} \int_{-\infty}^\infty  e^{-is\epsilon} g(is)  M_L(is) \ ds, \label{eq:parseval}
\end{align}
where 
\begin{equation}
    g(z)\triangleq z^{-1}(1+z)^{-1}
\end{equation} is the Laplace transform of the function $(1-e^{-x})^+$.  We can take any integration path parallel to the imaginary axis in the above integral as long as the real part $t$ of the line is greater than zero and belongs to the interior of $\mathcal{C}$, i.e.,
\begin{align}
     \delta_{L}(\epsilon)&= \frac{1}{2\pi i} \int_{t-i \infty}^{t+i\infty} e^{-z\epsilon}M_L(z)g(z) \ dz. \label{eq:parallel_path}
\end{align}
The fact that the integrals in \eqref{eq:parseval} and \eqref{eq:parallel_path} are equivalent is an immediate (and somewhat magical) consequence of Cauchy's integral theorem. However, an alternative way to derive this equivalence uses \emph{exponential tilting} of distributions; this method, while more cumbersome, provides valuable intuition that will also be used to derive precise bounds on the DP achievable for composed mechanisms. Exponential tilting is defined as follows.

\begin{definition}
The \emph{exponential tilting} with parameter $t\in \BR$ of a random variable $L$ having a finite CGF at $t$ is the random variable $\tilde{L}$ whose probability measure is given by 
\begin{equation}
    P_{\tilde{L}}(B) := \frac{1}{M_L(t)}\int_B e^{tu} \ dP_L(u)
\end{equation} 
for any Borel set $B$.
If $L$ has PDF $p_L$, then $\tilde{L}$ is given by its PDF $p_{\tilde{L}}(\ell)=e^{t\ell}p_L(\ell)/M_L(t)$.
\end{definition} 

A key feature of exponential tilting is that, for independent $L_i$, if $L=L_1+\cdots+L_n$, then $\tilL=\tilL_1+\cdots+\tilL_n$, where $\tilL_i$ is the exponential tilting of $L_i$, and again $\tilL_1,\ldots,\tilL_n$ are independent. This fact will be critical in deriving sharp bounds on composition curves, as the composition curve involves a PLRV that can be written as a sum of independent variables. We will return to this analysis in Section~\ref{sec:BE}, in which we derive bounds on the composition curve, but the analysis in this section applies to an arbitrary PLRV $L$.

The expectation of any function of $\tilL$ can be written in terms of $L$ as
\begin{equation}
    \bbE[f(\tilL)]=\frac{\bbE[e^{tL}f(L)]}{M_L(t)}.
\end{equation}
Thus, the MGF of the tilted variable $\tilL$ is given by
\begin{align}
    M_{\tilL}(z)&=
    \bbE[e^{z \tilL}]=\frac{\bbE[e^{tL} e^{zL}]}{M_L(t)}=\frac{M_L(t+z)}{M_L(t)}.
\end{align}
In addition, $K_L'(t)$ and $K_L''(t)$ can be interpreted as the expectation and variance respectively of $\tilL$ with parameter $t$. Similarly, expectations of functions of $L$ can be written in terms of $\tilL$ as
\begin{equation}
    \bbE[f(L)]=M_L(t) \,\bbE[e^{-t\tilL}f(\tilL)].
\end{equation}
Thus, we can write the privacy curve in terms of the tilted variable $\tilL$ as
\begin{align}
    \delta_{L}(\epsilon)&=
    \bbE\left[\left(1-e^{-(L-\epsilon)}\right)^+\right]
    \\
    &=M_L(t)\,\bbE\left[e^{-t\tilL}\left(1-e^{-(\tilL-\epsilon)}\right)^+\right].\label{eq:tilted_expectation}
\end{align}
Applying the Plancherel-Parseval identity to the expectation in \eqref{eq:tilted_expectation} gives
\begin{align}
    \delta_{L}(\epsilon)&=M_L(t) \frac{1}{2\pi} \int_{-\infty}^\infty e^{-(t+is)\epsilon} M_{\tilL}(is)g(t+is)ds
    \\&=\frac{1}{2\pi} \int_{-\infty}^\infty e^{-(t+is)\epsilon} M_{L}(t+is)g(t+is)ds\label{eq:fourier-tilt0}
    \\&=\frac{1}{2\pi i} \int_{t-i \infty}^{t+i\infty} 
    e^{-z\epsilon} M_L(z)g(z)dz, \label{eq:fourier-tilt}
\end{align}
where in \eqref{eq:fourier-tilt0} we have used the form of the MGF of $\tilL$. Therefore, we see that taking an integration path at real part $t$ is equivalent to exponential tilting with parameter $t$. This observation will be key to derive the {\sdp} accountant later in this section. First, however, we illustrate how exponential tilting can be used to re-derive known improvements to the moments accountant.

\subsection{Elementary Improvements of the Original Moments Accountant Method} \label{sec:ma}

The goal of this section is to motivate the method of steepest descent as a more accurate refinement to the moments accountant. We do this by showing that the moments accountant method proposed in~\cite{Abadi_MomentAccountant} is a simple but crude upper bound on~\eqref{eq:tilted_expectation}. We then re-derive two tighter bounds using the language of exponential tilting. Although these bounds already exist in the literature, the usage of exponential tilting to derive them is novel. 

 Define
\begin{equation}
\label{eq:fdefs}
    \bar{f}(x,t)\triangleq e^{-xt}\left(1-e^{-x}\right)^+~\mbox{and}~f(x,t)\triangleq e^{-xt}\left(1-e^{-x}\right).
\end{equation}
It follows directly from  \eqref{eq:tilted_expectation} that 
\begin{equation}
    \label{eq:MAvar}
    \delta_L(\epsilon) = e^{K_L(t)-\epsilon t}~ \bbE\left[\bar{f}\left(\tilL -\epsilon ,t\right) \right].
\end{equation}
The moments accountant bound can be  derived from \eqref{eq:MAvar} by observing that $\bar{f}(x,t)\leq 1$ for all $x\in \mathbb{R}$ and $t\geq0$ and, consequently, 
\begin{equation}
    \label{eq:MA}
     \delta_L(\epsilon) \leq \inf_{t\geq 0} \ \exp\left(K_L(t)-\epsilon t\right).
\end{equation}
The above expression is exactly the method proposed in \cite{Abadi_MomentAccountant}. The minimizing $t$ is given by 
\begin{equation}
    t_{\text{MA}} \define \left\lbrace \begin{array}{cl}
    \text{the unique solution to } K_L'(t)=\eps, & \text{ if } \BE[L]\le \eps, \\
    0, & \text{ if } \BE[L] > \eps.
    \end{array} \right.
\end{equation}
Note that the optimal tilting parameter $t_{\text{MA}}$ (in the case $\BE[L] \le \eps$) is the one for which the tilted PLRV satisfies $\bbE[\tilde{L}]=\epsilon$.

The bound $\bar{f}(x,t)\leq 1$ is rather crude, as illustrated in Fig.~\ref{fig:crudeapprox}. One simple improvement is to compute the maximum of $\bar{f}(\, \cdot\, ,t)$ for a fixed $t\geq 0$, rendering\footnote{The uniform bound in~\eqref{eq:fbar_max} is formally proved in Appendix~\ref{app:BE}.}
\begin{equation}\label{eq:fbar_max}
    \bar{f}(x, t)\leq t^t (t+1)^{-t-1} \qquad  \text{for all } x\in \mathbb{R},~t\geq 0.
\end{equation}
Since the right-hand side of \eqref{eq:fbar_max} is smaller than one except when $t=0$, we obtain a bound that is a strict tightening of \eqref{eq:MA}:
\begin{equation}
     \delta_L(\epsilon) \leq \inf_{t\geq 0} \  \exp\left(K_L(t)-\epsilon t+t\log t-(t+1) \log (t+1)\right).
     \label{eq:min_1}
\end{equation}
Here, the minimizing tilting $t$ is the solution of
\begin{equation}
    K'_L(t) = \epsilon+\log\left(\frac{t+1}{t}\right)
    \label{eq:tmin_1}
\end{equation}
if $\esup L > \eps$, and it is $t=\infty$ otherwise. 
This refinement of the moments accountant has been already proposed in \cite{asoodeh2020better} using a fundamentally different approach (see also \cite{canonne2020discrete, Balle2019HypothesisTI} for different proofs). 
We use this refined version of moments account, which is the latest version implemented in TensorFlow Privacy, in our numerical experiments (see the ``Baselines'' paragraph in Section~\ref{sec:numerical}).   

\begin{figure}[!t]
    \centering
    \includegraphics[width=0.7\textwidth]{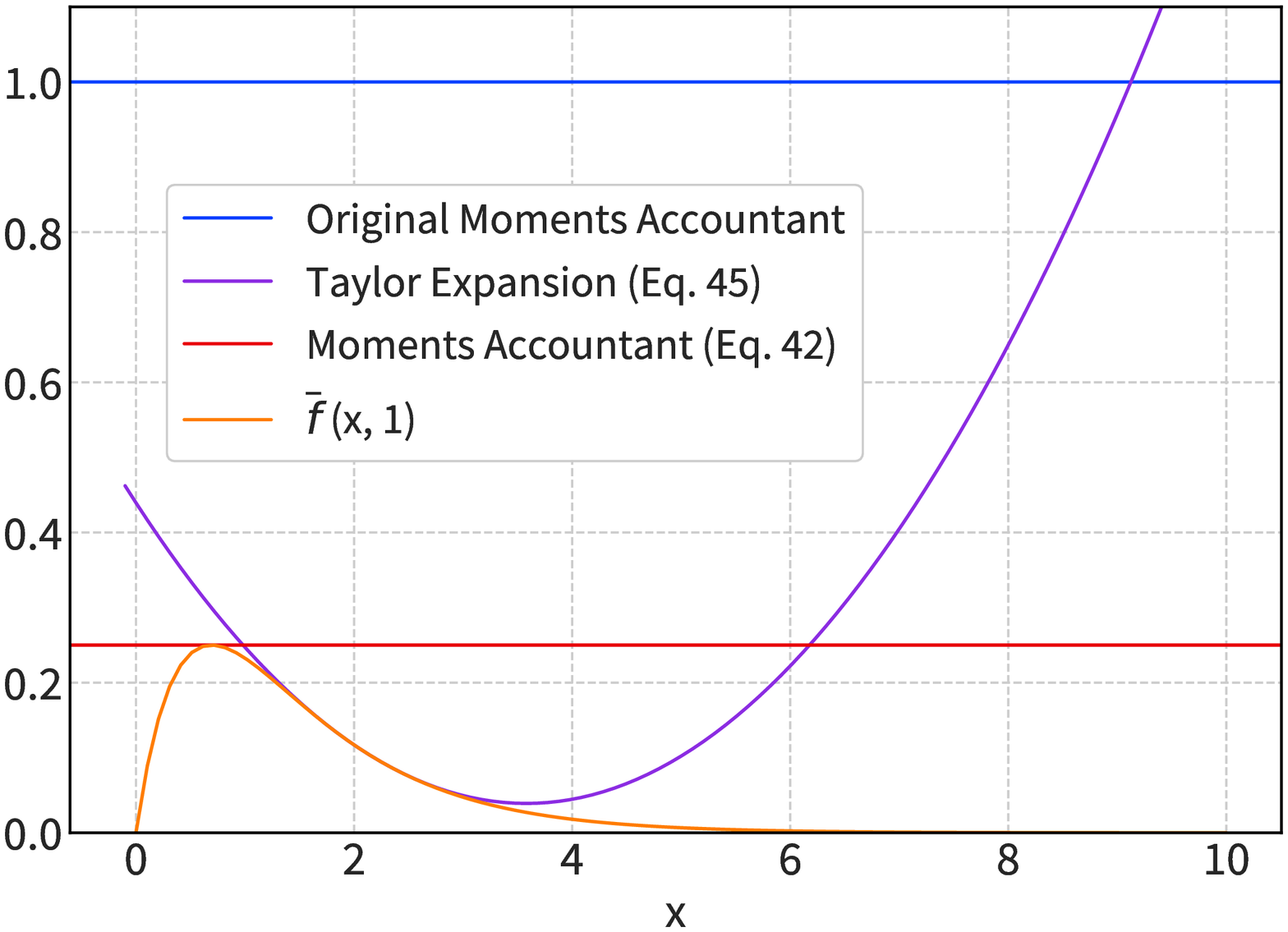}
    \caption{Comparison of $\bar{f}(x,1)$ with the various upper bounds in Section~\ref{sec:ma}.}
    \label{fig:crudeapprox}
\end{figure}

A second elementary improvement of \eqref{eq:MA} can be derived by further bounding $\bar{f}(x,t)$. It is a simple exercise to show that the series expansion of $\bar{f}(x,t)$ around $x_0(t) \triangleq 3\log\left(\frac{t+1}{t} \right)$ yields a quadratic upper bound that holds for any $t\geq 0$: 
\begin{equation}
    \bar{f}(x,t) \leq f(x_0(t),t) + \frac{\partial f}{\partial x}(x_0(t),t)(x-x_0(t)) + \frac{1}{2}\frac{\partial^2 f}{\partial x^2}(x_0(t),t)(x-x_0(t))^2,
\end{equation}
where $f$ is defined in \eqref{eq:fdefs}. By selecting $t$ such that the mean of $\tilL -\epsilon $ is exactly $x_0(t)$, we arrive at the bound
\begin{equation}
    \label{eq:taybound}
     \delta_L(\epsilon) \leq  \exp\left(K_L(t_0)-\epsilon t_0\right) \left(f\left(x_0(t_0),t_0\right)+t_0\left( \frac{t_0}{1+t_0}\right)^{1+3t_0} K_L''(t_0)\right),
\end{equation}
where $t_0\geq 0$ satisfies
\begin{equation}
    K_L'(t_0) = \epsilon + 3\log\left(\frac{t_0+1}{t_0} \right).
    \label{eq:t_taybound}
\end{equation}
Observe that the bound \eqref{eq:taybound} involves two key ingredients: (i) a ``good'' choice of the mean of the tilted random variable $\tilL$ (also needed in \eqref{eq:MA} and \eqref{eq:MAvar}), and (ii) the variance (second cumulant) of $\tilL$, given by $K''_L$. Using \eqref{eq:fourier-tilt} as a starting point, we demonstrate next how the method of steepest descent combines these two ingredients to derive a more accurate approximation of $\delta_L(\epsilon).$

\subsection{The Method of Steepest Descent and the Saddle-point Approximation}\label{sec:saddlepoint}

To derive the {\sdp} approximation, we rewrite \eqref{eq:parallel_path} once more as
\begin{equation}
     \delta_{L}(\epsilon)= \frac{1}{2\pi i} \int_{t-i \infty}^{t+i\infty} e^{F_\epsilon(z) }dz\label{eq::E_F}
\end{equation}
where we have defined the exponent as
\begin{equation}
    F_\epsilon(z)\triangleq K_L(z) -z\epsilon -\log z -\log(1+z) .
\end{equation}
The key in the {\sdp} method is the choice of $t$. The \emph{{\sdp}} is defined as the real value $t_0$ for which $F'_\epsilon(t_0)=0$. We assume that $\eps < \esup L$. Let $K_L|_{\BR}$ be the restriction of the CGF to the real axis. Then, $K_L|_{\BR}$ is convex, and thus, $F_\epsilon|_{\BR}$ is strictly convex. Thus, the minimum of $F_\eps$ over the reals is unique and satisfies $F'_\epsilon(t_0)=0$, i.e.,
\begin{equation}
    \label{eq:tau0}
    K_L'(t_0)=\epsilon+\frac{1}{t_0}+\frac{1}{t_0+1}.
\end{equation}
Observe that the original moments accountant aims to solve $K_L'(t)=\epsilon$ without the last two terms. The use of the term \emph{{\sdp}} can be justified by writing the second-order Taylor expansion of $F_\epsilon(z)$ around $t_0$:
\begin{equation}
    F_\epsilon(z)\approx F_\epsilon(t_0)+\frac{(z-t_0)^2}{2}F_\epsilon''(t_0).
\end{equation}
Since $F_\epsilon|_{\BR}$ is strictly convex, $F_\epsilon''(t_0)>0$, so along the real line, $F_\epsilon(z)$ has a minimum at $z=t_0$, whereas parallel to the imaginary axis, $F_\epsilon(z)$ has a maximum at $z=t_0$ (because in this case the coefficient $(z-t_0)^2$ becomes negative).

From this point, two closely related paths can be taken to approximate \eqref{eq::E_F}. The first approach is a direct application of the method of steepest descent, where $F_\epsilon$ is expanded around the {\sdp} $t_0$. The second approach expands \emph{only} $K_L$ around $t_0$, which is equivalent to deriving the Edgeworth series expansion of the distribution of $\tilL$ with tilting $t_0$. We outline both approaches here, but note that the second path (i.e., the expansion of $K_L$) is more amenable towards an approximation error analysis (see Section~\ref{sec:BE}), while the first path leads to better approximations in our numerical experiments.

\subsubsection{Method of Steepest Descent}\label{sec:steepest_descent}

For a fixed $\epsilon$, we next introduce the function $G_{\epsilon,t_0}$ that captures the difference between $F_\epsilon$ and the first two terms of its series expansion around $t_0$
\begin{equation}
    G_{\epsilon,t_0}(z) \triangleq F_\epsilon(t_0 +z) - F_\epsilon(t_0)-\frac{z^2}{2} F''_\epsilon(t_0).
\end{equation}
Then \eqref{eq::E_F} becomes
\begin{align}
    \label{eq::Gint}
     \delta_{L}(\epsilon)&= \frac{e^{F_\epsilon(t_0)}}{2\pi} \int_{- \infty}^{\infty} e^{-s^2F_\epsilon''(t_0) /2}e^{ G_{\epsilon,t_0}(is)}ds.
\end{align}
Observe that $G_{\epsilon,t_0}(0)=0,$ $G'_{\epsilon,t_0}(0)=0$, $G''_{\epsilon,t_0}(0)=0$, and $G^{(k)}_{\epsilon,t_0}(0)=F_\epsilon^{(k)}(t_0)$ for $k\geq 3$. Thus, we can expand $G_{\epsilon,t_0}$ via the series expansion
\begin{equation}
    G_{\epsilon,t_0}(z)=\sum_{k=3}^\infty \frac{F_{\epsilon}^{(k)}(t_0)}{k!} z^k,
\end{equation}
which holds on an open strip parallel to the imaginary axis centered at $t_0+i\BR$. Furthermore, $e^{G_{\epsilon,t_0}}$ can be written as
\begin{equation}
    \label{eq:bell}
    e^{G_{\epsilon,t_0}(z)}
    =1+\sum_{k=3}^\infty \frac{z^k}{k!} B_k(0,0,F_\epsilon^{(3)}(t_0),\ldots, F_\epsilon^{(k)}(t_0))
\end{equation}
where $B_k(x_1,\ldots,x_k)$ is the $k$th Bell polynomial. Applying this expansion to \eqref{eq::Gint} gives
\begin{align}
    &\delta_{L}(\epsilon)
    =\frac{e^{F_\epsilon(t_0)}}{2\pi}
    \int_{-\infty}^\infty e^{-s^2 F_\epsilon''(t_0)/2}\left(
    1+\sum_{k=3}^\infty \frac{(is)^k}{k!} B_k(0,0,F_\epsilon^{(3)}(t_0),\ldots, F_\epsilon^{(k)}(t_0))
    \right)ds
    \\&=\frac{e^{F_\epsilon(t_0)}}{2\pi}
    \left(\sqrt{\frac{2\pi}{F_\epsilon''(t_0)}}
    +\sum_{k\ge 3,k\text{ even}}\frac{i^k 2^{(k+1)/2} \Gamma((k+1)/2)}{k! F_\epsilon''(t_0)^{(k+1)/2}}B_k(0,0,F_\epsilon^{(3)}(t_0),\ldots, F_\epsilon^{(k)}(t_0))
    \right)
    \\&=\frac{e^{F_\epsilon(t_0)}}{2\pi}
    \left(\sqrt{\frac{2\pi}{F_\epsilon''(t_0)}}
    +\sum_{m=2}^\infty \frac{(-1)^m\sqrt{2\pi}}{2^m m! F_\epsilon''(t_0)^{m+1/2}}
    B_{2m}(0,0,F_\epsilon^{(3)}(t_0),\ldots, F_\epsilon^{(2m)}(t_0))
    \right)
    \\&=\frac{e^{F_\epsilon(t_0)}}{\sqrt{2\pi F_\epsilon''(t_0)}}
    \left(1+\sum_{m=2}^\infty \frac{(-1)^m}{2^m m! F_\epsilon''(t_0)^m}
    B_{2m}(0,0,F_\epsilon^{(3)}(t_0),\ldots, F_\epsilon^{(2m)}(t_0))
    \right).\label{eq:spa_series}
\end{align}

Based on the series expansion in \eqref{eq:spa_series}, we can derive various approximations depending on how many terms we keep. In particular, the simplest approximation is to drop the entire summation in~\eqref{eq:spa_series}, to obtain the approximation
\begin{equation}
    \delta_{L}(\epsilon)\approx
    \delta_{L, \, \text{SP-MSD}}^{(0)}(\epsilon)\define \frac{e^{F_\epsilon(t_0)}}{\sqrt{2\pi F_\epsilon''(t_0)}} = \frac{e^{K_{L}(t_0)-\epsilon t_0}}{\sqrt{2\pi}\sqrt{t_0(t_0+1)K_{L}''(t_0)+ t_0^2+(1+t_0)^2}}.
    \label{eq:true_saddlepoint}
\end{equation}
This approximation was presented in \eqref{SPA_approx1} where $L$ was a PLRV for  $\calM^{\circ n}$ given by $L_1+\dots+L_n$ for $L_i$'s being i.i.d.\ copies of the PLRV of $\calM$.   
If, in general, we assume that for a composed mechanism with $n$ compositions, $F_\epsilon(t)=O(n)$ (which follows from $K_{L}(t)$ also being $O(n)$), then a better approximation can be obtained by keeping the terms in~\eqref{eq:spa_series} that are $O(1/n)$; this yields 
\begin{equation}
    \delta_{L}(\epsilon) \approx 
    \delta_{L,\,\text{SP-MSD}}^{(1)}(\epsilon)\define
    \frac{e^{F_\epsilon(t_0)}}{\sqrt{2\pi F''_\epsilon(t_0)}}\left( 1+\frac{1}{8}\frac{F^{(4)}_\epsilon(t_0)}{F''_\epsilon(t_0)^2}-\frac{5}{24}\frac{F_\epsilon^{(3)}(t_0)^2}{F''_\epsilon(t_0)^3} \right). 
    \label{eq:true_saddlepoint_ncomp}
\end{equation}
Our numerical experiments indicate that the above approximation is surprisingly accurate,\footnote{As mentioned in~\cite[Page 49]{Huber_book}, the {\sdp} method often gives ``fantastically accurate'' estimates in practice.} and it can achieve a \emph{relative} error in $\delta_{L}(\epsilon)$ or $\eps_L(\delta)$ that is $<1\%$---even for a moderate number of compositions ($n\approx 100$; see Figure~\ref{fig:dve})! This approximation  involves computing higher order derivatives of $K_{L}(t)$, which is not difficult for most mechanisms used in practice. We provide implementation details in Section \ref{sec:numerical}.

\subsubsection{Edgeworth Expansion of the Tilted PLRV}\label{sec:CLT}

We return to \eqref{eq::E_F} with the choice of the {\sdp} $t_0$ satisfying~\eqref{eq:tau0}. The integral in \eqref{eq::E_F} can be rewritten as
\begin{align}
     \delta_{L}(\epsilon)&= \frac{1}{2\pi i} \int_{t_0-i \infty}^{t_0+i\infty} \frac{e^{-\epsilon z}e^{K_{L}(z)}}{z(z+1)}dz
     \label{eq::E_K}\\
     &=  \frac{1}{2\pi} \int_{-\infty}^\infty \frac{e^{- \epsilon (t_0+is)}e^{K_{L}(t_0+is)}}{(t_0+is)(t_0+is+1)}ds. \label{eq:deltaedge}
\end{align}
Unlike the previous subsection, here we approximate the integral via a series expansion of $e^{K_{L}(t_0+is)}$. This technique is equivalent to an Edgeworth expansion~\cite{hall2013bootstrap} of the distribution of $\tilde{L}$. The advantage of this approach is that any error bound for the truncation of the Edgworth series can be directly applied to bound the approximation error for $\delta_{L}(\epsilon)$. We make use of this fact in the next section for deriving an error bound for the {\sdp} accountant.

Picking the tilting parameter $t_0$ as the solution of \eqref{eq:tau0}, and assuming $\tilL$ has a density function $p_{\tilL}(x;t_0)$, the corresponding Edgeworth series expansion is~\cite[Eq.~(5)]{reid_saddlepoint_1988}
\begin{equation}
    \label{eq:edge_tilted}
    p_{\tilL}(x;t_0) \hspace{-2pt} = \hspace{-2pt} \frac{\phi\left(z(x,t_0)\right)}{\sqrt{K_{L}''(t_0)}}\left(1+\frac{\rho_3(t_0)}{6}h_3\left(z(x,t_0)\right)+\frac{\rho_4(t_0)}{24}h_4\left(z(x,t_0)\right)+\frac{\rho_3(t_0)^2}{72}h_6\left(z(x,t_0)\right) + \hspace{-1.3pt} \dots \right)
\end{equation}
where we denote by $\phi(x)$ the zero-mean unit-variance normal density, $\rho_3(t) \define K_{L}^{(3)}(t)/K_{L}''(t)^{3/2} $ and $\rho_4(t) \define K_{L}^{(4)}(t)/K_{L}''(t)^{2} $ are the normalized cumulants, 
\begin{equation}
    z(x,t_0) \define  \frac{x-K_{L}'(t_0)}{\sqrt{K_{L}''(t_0)}},
\end{equation}
and $h_k(z) \define (-1)^k\phi(z)^{-1}\phi^{(k)}(z)$ are the Hermite polynomials; e.g., $h_3(z)=z^3-3z$. The Edgeworth expansion approach delineated herein is different from what can be found in the DP literature~\cite{Wang2022Edgeworth}. Specifically, we apply the Edgeworth expansion on the \emph{tilted} random variable $\tilL$, whereas the approach in~\cite{Wang2022Edgeworth} uses the Edgeworth expansion of the non-tilted version $L$. This distinction can yield very different approximations, in the sense of non-asymptotic rate of convergence; see the comparison between our approach and the standard CLT in the discussion after Theorem~\ref{thm:err}.

Keeping only the first term of the above expansion is equivalent to approximating $\tilL$ by a Gaussian with the same mean and variance (i.e., the central limit theorem approximation); applying this approximation to \eqref{eq:MAvar} gives
\begin{equation}\label{eq:CLT_approx}
    \delta_{L}(\epsilon)\approx \delta_{L, \, \text{SP-CLT}}(\epsilon)\define e^{K_{L}(t_0)-\epsilon t_0} \ \bbE\left[\bar{f}(Z-\epsilon, t_0) \right],
\end{equation}
where $Z\sim \mathcal{N}\left(K_{L}'(t_0),K_{L}''(t_0) \right)$ and $\bar{f}$ is as defined in~\eqref{eq:fdefs}. We derive an error bound for this approximation in the next section.

While the two methods of approximation---the steepest descent, and the central limit theorem---lead to different approximations, as seen in \eqref{eq:true_saddlepoint} and \eqref{eq:CLT_approx}, these two approximations are closely related, as described by the following simple result.

\begin{proposition}
For any $t$ in the interior of $\calC$,
\begin{equation}\label{eq:BE_SP_upper_bound}
    \delta_{L,\, \textup{SP-CLT}}(\varepsilon)
    \le 
    \frac{e^{K_L(t)-\eps t}}{\sqrt{2\pi K_{L}''(t)}\ t(t+1)}
    =\frac{e^{F_\epsilon(t)}}{\sqrt{2\pi K_{L}''(t)}}.
\end{equation}
\end{proposition}
\begin{proof}
Let $Z\sim \mathcal{N}\left(K_L'(t),K_L''(t) \right)$ be the variable in the expectation in \eqref{eq:CLT_approx}. Its PDF is upper bounded by $p_Z(z)\le \frac{1}{\sqrt{2\pi K_{L}''(t)}}$. Thus
\begin{align}
    \bbE\left[e^{-t(Z-\eps)}\left(1-e^{-(Z-\eps)}\right)^+\right]
    &=\int_{\eps}^\infty p_Z(z)e^{-t(z-\eps)}\left(1-e^{-(z-\eps)}\right)dz
    \\&\le \frac{1}{\sqrt{2\pi K_{L}''(t)}} \int_{\eps}^\infty e^{-t(z-\eps)}\left(1-e^{-(z-\eps)}\right)dz
    \\&=\frac{1}{\sqrt{2\pi K_{L}''(t)}\ t(t+1)}.
\end{align}
Applying this bound to the definition of $\delta_{L,\, \text{SP-CLT}}(\varepsilon)$ in \eqref{eq:CLT_approx} completes the proof.
\end{proof}

Note that the only difference between the right-hand side of \eqref{eq:BE_SP_upper_bound} and $\dSDa{L}(\eps)$ is that the denominator involves $K_{L}''$ instead of $F_{\eps}''$.

\section{Error Bound Analysis by Applying Berry-Esseen to Tilts}\label{sec:BE}

While the approximations derived in the previous section are often very precise (see the numerical results in Section~\ref{sec:numerical}), they are merely \emph{approximations}, and do not provide any hard guarantees on the $(\epsilon,\delta)$-DP of a given mechanism. In this section, we derive upper and lower bounds on the achieved privacy compared to the approximation $\delta_{L^{(n)},\,\text{SP-CLT}}$ as given in~\eqref{eq:CLT_approx}. These bounds are derived by applying the Berry-Esseen theorem to~\eqref{eq:tilted_expectation}.

The following setup is fixed throughout this section. We consider the adaptive composition of $n$ mechanisms $\calM^{(n)} = \calM_1 \circ \cdots \circ \calM_n$, and assume that each constituent mechanism has a PLRV $L_j$ for which $\delta_{\calM_j}=\delta_{L_j}$. The composition curve $\bar{\delta}_{\calM^{(n)}}$ is then equal to $\delta_{L^{(n)}}$, where $L^{(n)}=L_1+\cdots+L_n$, and the $L_j$ are independent:
\begin{equation}
    \delta_{\calM^{(n)}} \le \bar{\delta}_{\calM^{(n)}} = \delta_{L^{(n)}}.
\end{equation}
Since the $L_j$ are independent, the CGF of $L^{(n)}$ is given by the sum
\begin{equation}
    K_{L^{(n)}}(t) \define \log \BE\left[ e^{tL^{(n)}} \right] = K_{L_1}(t) + \cdots+ K_{L_n}(t)
\end{equation}
where $K_j(t)=\log \bbE[e^{tL_j}]$ is the CGF of $L_j$. We also let 
\begin{equation}\label{eq:Pt_def}
\uP_t^{(n)} \define \sum_{j=1}^n \BE\left[ \left|\tilde{L}_j - \BE[\tilde{L}_j] \right|^3 \right], 
\end{equation}
where $\tilde{L}_j$ is the exponential tilting of $L_j$ with parameter $t$. Recall that the exponential tilting of $L^{(n)}$ with parameter $t$ is $\tilL^{(n)} = \tilL_1+\cdots+\tilL_n$, and the $\tilL_j$ are independent.

\subsection{Error Bounds for Arbitrary Tilts}

The following theorem gives error bounds on the approximation $\delta_{L^{(n)},\, \text{SP-CLT}}$.

\begin{theorem}\label{thm:BE}
For any $t$ in the interior of $\calC$, and any $\eps\ge 0$, there is a $\zeta\in[-1,1]$ such that
\begin{equation}\label{eq:BE_bounds}
\delta_{L^{(n)}}(\varepsilon)
=\exp\left(K_{L^{(n)}}(t)-\eps t\right)
\left(
\bbE\left[ e^{-t(Z-\varepsilon)}\left(1-e^{-(Z-\varepsilon)}\right)^+\right]
+\zeta \frac{t^t}{(1+t)^{1+t}}\cdot \frac{1.12\, \uP_t^{(n)}}{K_{L^{(n)}}''(t)^{3/2}}
\right),
\end{equation}
where $Z$ is Gaussian with mean $K_{L^{(n)}}'(t)$ and variance $K_{L^{(n)}}''(t)$.
\end{theorem}
\begin{proof}
See Appendix~\ref{app:BE}.
\end{proof}

Note that, omitting the $\zeta$ term in the right-hand side of \eqref{eq:BE_bounds} gives exactly $\delta_{L^{(n)},\,\text{SP-CLT}}$ (with the tilt $t$ being freely chosen from the interior of $\calC$). Thus, Theorem~\ref{thm:BE} can be equivalently written
\begin{equation}
    \left|\delta_{L^{(n)}}(\varepsilon)-\delta_{L^{(n)},\,\text{SP-CLT}}(\varepsilon)\right|\le \text{err}_{\text{SP}}(\eps)
\end{equation}
where we have the error term\footnote{To reduce cluttering, we suppress the dependence on $n$ and $t$ from the notation $\errt(\eps)$.}
\begin{align}
    \text{err}_{\text{SP}}(\eps)&
    \define
    \exp\left(K_{L^{(n)}}(t)-\eps t\right)\frac{t^t}{(1+t)^{1+t}}\cdot \frac{1.12\, \uP_t^{(n)}}{K_{L^{(n)}}''(t)^{3/2}}. \label{eq:berry_esseen_delta_error}
\end{align}

\begin{remark}
While Theorem~\ref{thm:BE} holds for any positive value of $t$ around which the MGF is finite, a natural choice of $t$ is to use the {\sdp} $t_0$ itself, defined as the solution to \eqref{eq:tau0}.
\end{remark}

The following definitions will be used to express the approximation function  $\delta_{L^{(n)},\,\text{SP-CLT}}(\varepsilon)$ without the Gaussian expectation. Proposition~\ref{prop:gaussian_computation} below allows for computing the approximation function in finite time (given $t$, $K_{L^{(n)}}(t)$, $K_{L^{(n)}}'(t)$, and $K_{L^{(n)}}''(t)$).

\begin{definition}
The (Gaussian) $Q$-function is defined by $Q(z) \define 1 - \Phi(z)$. We define the function $q:\BR \to (0,\infty)$ by $q(z) \define Q(z)\cdot \sqrt{2\pi} \ e^{z^2/2}$.
\end{definition}
\begin{remark}
The function $q$ is closely related to the \emph{Mills' ratio} $\textup{M}(z)\define e^{z^2}\int_z^\infty e^{-u^2} \, du$, and from known inequalities on the $Q$ function one may infer that 
\begin{equation}
    \frac{2}{z+\sqrt{z^2+4}} < q(z) \le \frac{2}{z+\sqrt{z^2+8/\pi}}
\end{equation}
for all $z\ge 0$~\cite[Section 7.8]{NIST}. In particular, one obtains $q(z)< \min(1/z,\sqrt{\pi/2})$ for all $z> 0$, and $q(z)\sim 1/z$ as $z\to \infty$.
\end{remark}
\begin{proposition}\label{prop:gaussian_computation}
For any $t$ in the interior of $\calC$ and $\epsilon\ge 0$, let
\begin{align}
    \gamma &\define \frac{K_{L^{(n)}}'(t)-\varepsilon}{\sqrt{K_{L^{(n)}}''(t)}},&
    \alpha &\define \sqrt{K_{L^{(n)}}''(t)} \ t - \gamma,&
    \beta &\define \sqrt{K_{L^{(n)}}''(t)} \ (t+1) - \gamma.
\end{align}
Then
\begin{equation}
    \delta_{L^{(n)},\, \textup{SP-CLT}}(\varepsilon) = 
    \exp\left({K_{L^{(n)}}(t)-\varepsilon t-\gamma^2/2} \right)
    \frac{q(\alpha)-q(\beta)}{\sqrt{2\pi}} .
    \label{eq:berry_esseen_delta}
\end{equation}
\end{proposition}
\begin{proof}
See Appendix~\ref{app:gaussian_computation}.
\end{proof}

Next, we study the rate of decay of $\errt$ as the number of compositions $n$ grows without bound, and when the tilt is chosen to be the {\sdp} $t_0$.

\subsection{Decay Rate of the Approximation Error for the Saddle-point Choice of Tilt}

We show that the error rate in approximating $\delta_{L^{(n)}}$ by $\dCL{L^{(n)}}$ decays as $1/\sqrt{n}$. Further, we characterize the constant term in this decay rate. To illustrate the advantage of our approach, we compare such a decay rate with what might be achieved applying CLT directly to approximate $\delta_{L^{(n)}}$ without performing exponential tilting on $L^{(n)}$ beforehand.

The results of this section hold under the following assumption, which we assume throughout.

\begin{assumption} \label{assumption:KL V P}
There are constants $\KL,\V>0$, and $\uP$ such that, for any $t=o(n^{-1/3})$, we have the limit $(\BE[\tilL^{(n)}],\sigma_{\tilL^{(n)}}^2,\uP_t^{(n)})/n \to (\KL,\V,\uP)$ as $n\to \infty$.
\end{assumption}
\begin{remark}
Equivalently, the above assumption postulates that 
\begin{align}
    \frac{K_{L_1}'(c/\sqrt{n}) + \cdots + K_{L_n}'(c/\sqrt{n})}{n} &\to \KL \\
    \frac{K_{L_1}''(c/\sqrt{n}) + \cdots + K_{L_n}''(c/\sqrt{n})}{n} &\to \V \\
    \frac{\BE[ |\tilde{L}_1 - \BE[K_{L_1}'(c/\sqrt{n})] |^3 ] + \cdots + \BE[ |\tilde{L}_n - \BE[K_{L_n}'(c/\sqrt{n})] |^3 ]}{n} &\to \uP, \label{ass:P}
\end{align}
where each $\tilL_j$ in~\eqref{ass:P} is a tilted version of $L_j$ by $t=c/\sqrt{n}$, and $c=o(n^{1/6})$. If $\calM^{(n)}$ is an $n$-fold composition (i.e., the $\calM_j$ are identical), then the above conditions are immediately satisfied, with $\KL=K_{L_1}'(0)=\BE[L_1]$, $\V=K_{L_1}''(0)=\sigma_{L_1}^2$, and $\uP = \BE[|L_1-\KL|^3]$.
\end{remark}

Our main error decay-rate result, stated in Theorem~\ref{thm:err} below, shows that the error in approximating $\delta_{L^{(n)}}(\veps)$ by $\dCL{L^{(n)}}(\eps)$ decays roughly at least as fast as $1/(\sqrt{n}e^{b^2/2})$ for ``interesting'' values of $\veps$. More precisely, we only consider $\veps = \BE[L^{(n)}]+b\sigma_{L^{(n)}}$ for $0<b\ll n^{1/6}$. We explain this choice of regime below in Theorem~\ref{thm:eps}, where we show that $\delta_{L^{(n)}}(\BE[L^{(n)}]-\Phi^{-1}(\delta)\sigma_{L^{(n)}}) \to \delta$ for any fixed level $\delta \in (0,1)$; e.g., the value of $\delta_{L^{(n)}}(\veps)$ is close to $10^{-10}$ if and only if the value of $\veps$ is around $\BE[L^{(n)}]+6.4\sigma_{L^{(n)}}$ for all large $n$. Thus, if one hopes to have a small value of $\delta$, the only viable values of $\veps$, in the regime of high $n$, are those that deviate from $\BE[L^{(n)}]$ by a small multiple of $\sigma_{L^{(n)}}$.

\paragraph{Error decay rate.} The decay rate of $\errt$ is characterized in the following theorem. Recall that the {\sdp} is the unique $t_0>0$ satisfying
\begin{equation}
    K_{L^{(n)}}'(t_0)=\epsilon+\frac{1}{t_0}+\frac{1}{t_0+1}.
\end{equation}

\begin{theorem} \label{thm:err}
For $\veps = \BE[L^{(n)}] + b \sigma_{L^{(n)}}$, where $b>0$ satisfies $b=o(n^{1/6})$, with the choice of tilt being the {\sdp} $t_0$, we have the asymptotic
\begin{equation}
    \errt(\veps) \sim \frac{1.12 \sqrt{e} \ \uP}{\V^{3/2} \cdot C(b)^\tau \cdot \sqrt{n}},
\end{equation}
where $\tau < 1$ satisfies $\tau \to 1$, and we define the term $C(b)$ by
\begin{equation}
    C(b) \define \exp\left( (b^2+b\sqrt{b^2+4})/4 \right).
\end{equation}
\end{theorem}
\begin{proof}
See Appendix~\ref{app:err}.
\end{proof}
\begin{remark}
In fact, the term $\tau$ can be given in terms of the {\sdp} $t_0$ as follows. Writing the finitary version of the asymptotic shown in Proposition~\ref{lem:t} below as
\begin{equation}
    t_0 = \tau_0 \cdot \frac{b+\sqrt{b^2+4}}{2\sigma_{L^{(n)}}},
\end{equation}
where $\tau_0>0$ is such that $\tau_0\to 1$, we have that
\begin{equation}
    \tau = (2-\tau_0)\tau_0.
\end{equation}
\end{remark}

\paragraph{Comparison with standard CLT.} To illustrate the advantage of our tilting approach, we compare the asymptotic behavior of the error in Theorem~\ref{thm:err} to that obtainable from non-tilted Berry-Esseen. By the Berry-Esseen theorem, we have for $Z\sim \calN(\BE[L^{(n)}], \sigma_{L^{(n)}}^2)$ that\footnote{Note that $Z$ is not necessarily a PLRV associated to a Gaussian mechanism, since in general $\sigma_{L^{(n)}}^2 \neq 2\BE[L^{(n)}]$.}
\begin{align} 
    \delta_{L^{(n)}}(\varepsilon) &= \BE\left[ \left( 1 - e^{-(L^{(n)}-\varepsilon)} \right)^+ \right] = \int_0^1 \BP\left[ L^{(n)} > \varepsilon - \log(1-u) \right] \, du = \delta_{Z}(\varepsilon) + \theta \cdot  \frac{0.56 \ \uP_0^{(n)}}{ \sigma_{L^{(n)}}^{3/2} } \label{eq:BE no tilt}
\end{align}
where $|\theta|\le 1$. Under our setup (in particular, Assumption~\ref{assumption:KL V P}), the error term in the standard Berry-Esseen approach shown above satisfies
\begin{equation}
    \mathrm{err}_{\text{Standard}}(\veps) \define  \frac{0.56 \ \uP_0^{(n)}}{ \sigma_{L^{(n)}}^{3/2} } \sim \frac{0.56 \ \uP}{\V^{3/2} \cdot \sqrt{n}}.
\end{equation}
Thus, the improvement our approach yields is asymptotically 
\begin{equation}
    \frac{\errt(\veps)}{\mathrm{err}_{\text{Standard}}(\veps)} \sim \frac{2\sqrt{e}}{C(b)^\tau}.
\end{equation}
Even for moderate values of $b$, the above ratio is very small (recall that we denote $\veps = \BE[L^{(n)}] + b \sigma_{L^{(n)}}$). For example, if $b\approx 6.4$ (so $\delta \approx 10^{-10}$ in the limit; see Theorem~\ref{thm:eps} below), we obtain the limit of the ratio
\begin{equation}
    \lim_{n\to \infty} \ \frac{\errt(\veps)}{\mathrm{err}_{\text{Standard}}(\veps)} \approx 3\times 10^{-9}.
\end{equation}

In addition, in the complementary regime of $\delta \to 0$, e.g., when $\eps = \BE[L^{(n)}]+b\sigma_{L^{(n)}}$ with $b \ge \sqrt{\log n}$ (and still $b=o(n^{1/6})$), one has that the error term in the standard CLT \emph{dominates} the approximation of $\delta$:
\begin{equation}
    \delta_Z(\eps) = o\left(\mathrm{err}_{\text{Standard}}(\veps)\right).
\end{equation}
In contrast, in the same regime, our error term $\errt(\eps)$ is always smaller than the approximation itself, i.e., 
\begin{equation}
    \errt(\eps) = o\left( \dCL{L^{(n)}}(\eps) \right).
\end{equation}

\paragraph{Saddle-point asymptotic.} The essential ingredient for the proof of Theorem~\ref{thm:err} is the following characterization of the {\sdp}, which could be of independent interest.

\begin{proposition} \label{lem:t}
For $\veps = \BE[L^{(n)}] + b \sigma_{L^{(n)}}$, where $b>0$ satisfies $b=o(n^{1/6})$,  the {\sdp} satisfies the asymptotic relation
\begin{equation}
    t_0 \sim \frac{b+\sqrt{b^2+4}}{2\sigma_{L^{(n)}}}.
\end{equation}
\end{proposition}
\begin{proof}
See Appendix~\ref{app:t}.
\end{proof}

\paragraph{The high-composition regime.} As mentioned before Theorem~\ref{thm:err}, we only consider values of $\veps$ that deviate from $\BE[L^{(n)}]$ by a small multiple of $\sigma_{L^{(n)}}$ since that is necessary for $\delta_{L^{(n)}}(\veps)$ to be guaranteed to converge to a reasonable value of $\delta$. The following result shows this fact formally.

\begin{theorem} \label{thm:eps}
For any $\delta \in (0,1/2)$, we have that
\begin{equation}
    \delta_{L^{(n)}}(\BE[L^{(n)}] - \Phi^{-1}(\delta) \sigma_{L^{(n)}}) \to \delta,
\end{equation}
i.e.,
\begin{equation}
    \frac{\veps_{L^{(n)}}(\delta) - \BE[L^{(n)}]}{\sigma_{L^{(n)}}} \to - \Phi^{-1}(\delta).
\end{equation}
\end{theorem}
\begin{proof}
We may show this result using the standard Berry-Esseen approach, as in~\eqref{eq:BE no tilt}. A direct computation yields that, for any $\veps\ge \BE[L^{(n)}]$, with $Z\sim \calN(\BE[L^{(n)}],\sigma_{L^{(n)}}^2)$,
\begin{align}
    \delta_{Z}(\varepsilon) &= \Phi\left( \frac{\BE[L^{(n)}] - \varepsilon}{\sigma_{L^{(n)}}} \right) - e^{\varepsilon - \BE[L^{(n)}] + \sigma_{L^{(n)}}^2/2} \ \Phi\left( \frac{\BE[L^{(n)}] - \sigma_{L^{(n)}}^2 - \varepsilon}{\sigma_{L^{(n)}}} \right).
\end{align}
Taking $\veps = \BE[L^{(n)}] - \Phi^{-1}(\delta) \sigma_{L^{(n)}}$ yields the limit $\delta_Z(\veps) \to \delta$. By~\eqref{eq:BE no tilt}, we obtain $\delta_{L^{(n)}}(\veps) \to \delta$ too.
\end{proof}
\begin{remark}
This result shows, in particular, that the mean and variance for a PLRV are unique for a fixed mechanism $\calM$ characterized by $L$. In other words, if $L$ and $L'$ are both PLRVs for $\calM$ such that $\delta_{\calM}=\delta_L = \delta_{L'}$, then we must have that $\BE[L]=\BE[L']$ and $\sigma_L=\sigma_{L'}$. This fact was first observed in~\cite{sommer2019privacy}.  
\end{remark}

\section{Experiments}\label{sec:numerical}

We benchmark SPA against state-of-the-art accounting methods via numerical experiments. We focus on accounting for privacy of the DP-SGD algorithm~\cite{Abadi_MomentAccountant} under composition. In particular, the mechanism we are accounting for is the $n$-fold composition $\calM_{\lambda}^{\circ n}$, where $\calM_\lambda$ is the $\lambda$-subsampled Gaussian mechanism. 

We will either fix $\delta$ and approximate $\eps$ as a function of $n$, or fix $n$ and approximate $\delta$ as a function of $\eps$. These two experiments, in short, provide evidence that SPA outperforms other constant-time composition accountants while achieving comparable accuracy to the state-of-art FFT accountant. While the large-deviation based moments accountant overreports the privacy budget, and the CLT-based GDP accountant underreports the privacy budget, SPA combines the two mathematical approaches (large-deviation and CLT) to achieve a reasonable approximation of the privacy budget. 

Further, while the FFT based approach proposed in~\cite{Numerical_compositionDP} can achieve arbitrary accuracy, it does not enjoy the fast time complexity SPA does. Moreover, the necessary discretization of the PLRV in the FFT approach makes floating point errors a central concern when $\delta <  10^{-15} \times \text{discretization points}$. In our experiments, the number of discretization points is on the order of $10^{4}$, meaning $\delta$ cannot be computed below roughly $10^{-11}$. In contrast, the simplicity of the SPA allows us to straightforwardly implement a floating point stable algorithm.

\begin{algorithm}[!t]
   \caption{\textbf{:} \textsc{SaddlePointAccountant} (SPA)}
   \label{alg:main}
\begin{algorithmic}[1]
   \State {\bfseries Input:} Pairs of tightly dominating distributions $(P_1,Q_1), \ldots, (P_n,Q_n)$ for mechanisms $\calM_1,\cdots,\calM_n$, and a finite set $\calE \subset [0,\infty)$ (values of $\eps$).
   \State {\bfseries Output:} Three approximations $ \dSDa{L^{(n)}}, \dSDb{L^{(n)}}$, and $\dCL{L^{(n)}}$ of the composition curve $\bar{\delta}_{\calM_1\circ \cdots \circ \calM_n}$, and an error bound  so that $|\bar{\delta}_{\calM_1\circ \cdots \circ \calM_n}(\eps)-\dCL{L^{(n)}}(\eps)|\le \errt(\eps)$.
   \vspace{2mm}
   \State $L_j \leftarrow \log \frac{\textup{d} P_j}{\textup{d} Q_j}(X_j)$ where $X_j \sim P_j$ $\hfill~ j\in [n]$
   \vspace{0.5mm}
   \State $K_{L_j}(t) \leftarrow \log \BE\left[ e^{tL_j} \right] \hfill~ j\in [n]$
   \vspace{0.5mm}
   \State $L^{(n)} \leftarrow L_1+\cdots+L_n$
   \vspace{0.5mm}
   \State $K_{L^{(n)}} \leftarrow K_{L_1} + \cdots + K_{L_n}$
   \vspace{1.5mm}
   \For{$\eps \in \calE$}
   \State $t_0 \leftarrow$ the unique positive solution to ${\displaystyle K'_{L^{(n)}}(t_0) = \eps + \frac{1}{t_0} + \frac{1}{t_0+1}}$
   \vspace{1mm}
   \State $F_\epsilon(t) \leftarrow K_{L^{(n)}} (t) - \epsilon t - \log t - \log (t+1)$
   \vspace{1mm}
   \State ${\displaystyle \dSDa{L^{(n)}}(\eps) \leftarrow \frac{e^{F_\epsilon(t_0)}}{\sqrt{2\pi F''_\epsilon(t_0)}}}$
   \vspace{1mm}
   \State ${\displaystyle \dSDb{L^{(n)}}(\eps) \leftarrow \frac{e^{F_\epsilon(t_0)}}{\sqrt{2\pi F''_\epsilon(t_0)}}\left( 1+\frac{1}{8}\frac{F^{(4)}_\epsilon(t_0)}{F''_\epsilon(t_0)^2}-\frac{5}{24}\frac{F_\epsilon^{(3)}(t_0)^2}{F''_\epsilon(t_0)^3} \right)}$
   \vspace{1mm}
   \State $\left( \gamma, \alpha , \beta \right) \leftarrow \left(  \frac{K_{L^{(n)}}'(t_0)-\varepsilon}{\sqrt{K_{L^{(n)}}''(t_0)}}, \sqrt{K_{L^{(n)}}''(t_0)} \ t_0 - \gamma, \sqrt{K_{L^{(n)}}''(t_0)} \ (t_0+1) - \gamma \right)$
   \vspace{1mm}
   \State ${\displaystyle \dCL{L^{(n)}}(\eps) \leftarrow \exp\left({K_{L^{(n)}}(t_0)-\varepsilon t_0-\gamma^2/2} \right)
    \frac{q(\alpha)/\alpha-q(\beta)/\beta}{\sqrt{2\pi}}}$
    \vspace{1mm}
    \State $\tilL_j \leftarrow$ exponential tilting of $L_j$ with parameter $t_0 \hfill~ j\in [n]$
    \vspace{1mm}
    \State ${\displaystyle \uP_{t_0}^{(n)} \leftarrow \sum_{j\in [n]} \BE\left[ \left| \tilL_j - K_{L_j}'(t_0) \right|^3 \right]}$ 
    \State ${\displaystyle \errt(\eps) \leftarrow \exp\left( K_{L^{(n)}}(t_0) - \eps t_0 \right) \frac{t_0^{t_0}}{(1+t_0)^{1+t_0}} \cdot \frac{1.12 \ \uP_{t_0}^{(n)}}{K_{L^{(n)}}''(t_0)^{3/2}}}$
    \vspace{1mm}
    \EndFor
    \vspace{1mm}
    \State {\bfseries Return:} ${\displaystyle \dSDa{L^{(n)}}, \  \dSDb{L^{(n)}}, \ \dCL{L^{(n)}}, \ \errt}$.
\end{algorithmic}
\end{algorithm}

\begin{algorithm}[!t]
   \caption{\textbf{:} SPA for the subsampled Gaussian mechanism}
   \label{alg:gaussian_algo}
\begin{algorithmic}[1]
   \State {\bfseries Input:} Noise variance $\sigma^2$, $\ell_2$ sensitivity $s$, number of compositions $n$, and subsampling rate $\lambda \in [0,1]$.
   \State {\bfseries Output:} Three approximations $ \dSDa{L^{(n)}}, \dSDb{L^{(n)}}$, and $\dCL{L^{(n)}}$ of the composition curve $\bar{\delta}_{\calM_\lambda^{\circ n}}$ where $\calM_\lambda^{\circ n}$ is the $n$-fold adaptive composition of $\calM_\lambda$, a subsampled Gaussian mechanism with the given parameters.
   \vspace{2mm}
   \State $P_j \leftarrow (1-\lambda) \ \calN(0, \sigma^2) +\lambda \ \calN(s,\sigma^2) \hfill~ j\in[n]$
   \State $Q_j \leftarrow \calN(0, \sigma^2)$ {\hfill~ $j\in [n]$ \ignorespaces}
   \vspace{1mm}
   \State {\bfseries Return:} ${\displaystyle \textsc{SaddlePointAccountant}\left( \{(P_j,Q_j)\}_{j\in [n]} \right) }$
\end{algorithmic}
\end{algorithm}

\paragraph{Experimental setup.} Recall that all accountants presented herein aim at approximating the composition curve $\bar{\delta}_{\calM^{\circ n}}$ (see Definition~\ref{Def:CompositionCurve}). Equivalently, if $L$ is a PLRV for the subsampled Gaussian mechanism $\calM_\lambda$, then $\bar{\delta}_{\calM^{\circ n}} = \delta_{L^{(n)}}$ where $L^{(n)}=L_1+\cdots+L_n$ for i.i.d. $L_1,\cdots,L_n$. By Lemma~\ref{lem:subsampling}, we may take $L_i, i=1,\ldots,n$ to be
\begin{equation}
    L_i = \log \left( 1 - \lambda + \lambda e^{s(2X-s)/(2\sigma^2)} \right), \qquad X\sim (1-\lambda)\calN(0,\sigma^2) + \lambda \calN(s,\sigma^2),
\end{equation}
where $s$ is a prespecified $\ell_2$ sensitivity, and $\sigma^2$ is the variance of the underlying Gaussian mechanism. Without loss of generality, we fix $s=1$. Our experiments are set in either of two scenarios: approximating $\eps \mapsto \delta_{L^{(n)}}(\eps)$ (for fixed $n$), or $n\mapsto \eps_{L^{(n)}}(\delta)$ (for fixed $\delta$), where $\eps_{L^{(n)}}$ is the inverse of $\delta_{L^{(n)}}$ as given by Definition~\ref{def:PLRV}. We plot five different ways of computing the privacy budget: our proposed SPA, three other baselines (explained below), and an almost-exact (yet computationally expensive) computation of the privacy parameters (labeled ``Truth'' in the plots). In addition to them clearly helping in assessing the correctness of the various accountants, the ``Truth'' curves also allow us to zoom in and distinguish the relative error offered by SPA in contrast to the other accountants. The relative error is define as follows: if $\hat{\delta}$ is an estimate of $\delta_{L^{(n)}}(\eps)$, then its relative error is defined as $r=|\hat{\delta}-\delta_{L^{(n)}}(\eps)|/\delta_{L^{(n)}}(\eps)$, so $\hat{\delta} = \delta_{L^{(n)}}(\eps)\left(1+\eta r\right)$ for $\eta\in \{\pm 1\}$; the relative error when approximating $\eps_{L^{(n)}}(\delta)$ is defined similarly. When we say that ``the relative error is $<u\%$,'' we mean that $r<u/100$. 

\paragraph{Truth Curve.} The ``Truth'' curves are computed as follows. To compute $\delta_{L^{(n)}}(\eps)$ for fixed $\eps$ and $n$, we use the integral representation in~\eqref{eq:fourier-tilt}, and perform Gaussian quadrature to 50 digits of precision using the \texttt{mpmath}~\cite{mpmath} package in \texttt{python}; then, for computing $\eps_{L^{(n)}}(\delta)$, we perform a binary search. We reiterate that producing the ``Truth'' curves is computationally very expensive, and we perform it only for sake of producing the relative error curves.

\paragraph{The {\sdp} accountant.} The approximations of $\delta_{L^{(n)}}$ that we propose are given by $\dCL{L^{(n)}}$, $\dSDa{L^{(n)}}$, and $\dSDb{L^{(n)}}$, as given by their formulas in~\eqref{SPA_approx2},~\eqref{SPA_approx1}, and~\eqref{eq:true_saddlepoint_ncomp}, respectively; these are referred to in the plots by the labels ``SP-CLT,'' ``SP-MSD0,'' and ``SP-MSD1,'' respectively. We summarize the procedure for calculating the SPA in Algorithm~\ref{alg:main}. We also instantiate Algorithm~\ref{alg:main} for the DP-SGD setting in Algorithm~\ref{alg:gaussian_algo}, whose validity is guaranteed by Lemma~\ref{lem:subsampling}.

\paragraph{Baselines.} \label{baselines} We compare SPA against three baseline DP accountants: the state-of-the-art FFT-based approach proposed in~\cite{Numerical_compositionDP}; the moments accountant\footnote{We use the latest implementation of the moments accountant. As in \href{https://github.com/google/differential-privacy/blob/main/python/dp_accounting/rdp/rdp_privacy_accountant.py}{this github link}, this implementation uses the improved bounds discussed at the beginning of Section~\ref{sec:ma}.}~\cite{Abadi_MomentAccountant}; and the GDP accountant\footnote{We use the implementation at~\href{https://github.com/tensorflow/privacy/blob/master/tensorflow_privacy/privacy/analysis/gdp_accountant.py}{this github link}.}~\cite{GaussianDP}. We use the implementations for the moments accountant and the GDP accountant exactly as provided in their github repositories. For the FFT accountant in~\cite{Numerical_compositionDP}, we use throughout the experiments the parameters $\eps_{\error}=0.1$, since it is the one used in~\cite{Numerical_compositionDP}, and $\delta_{\error}=10^{-10}$, since it is the minimal order allowed for this parameter as per the discussion in~\cite[Appendix B]{Numerical_compositionDP}.

\begin{figure}[ht]
    \centering
    \begin{minipage}{0.5\textwidth}
        \centering
        \includegraphics[width=1\textwidth]{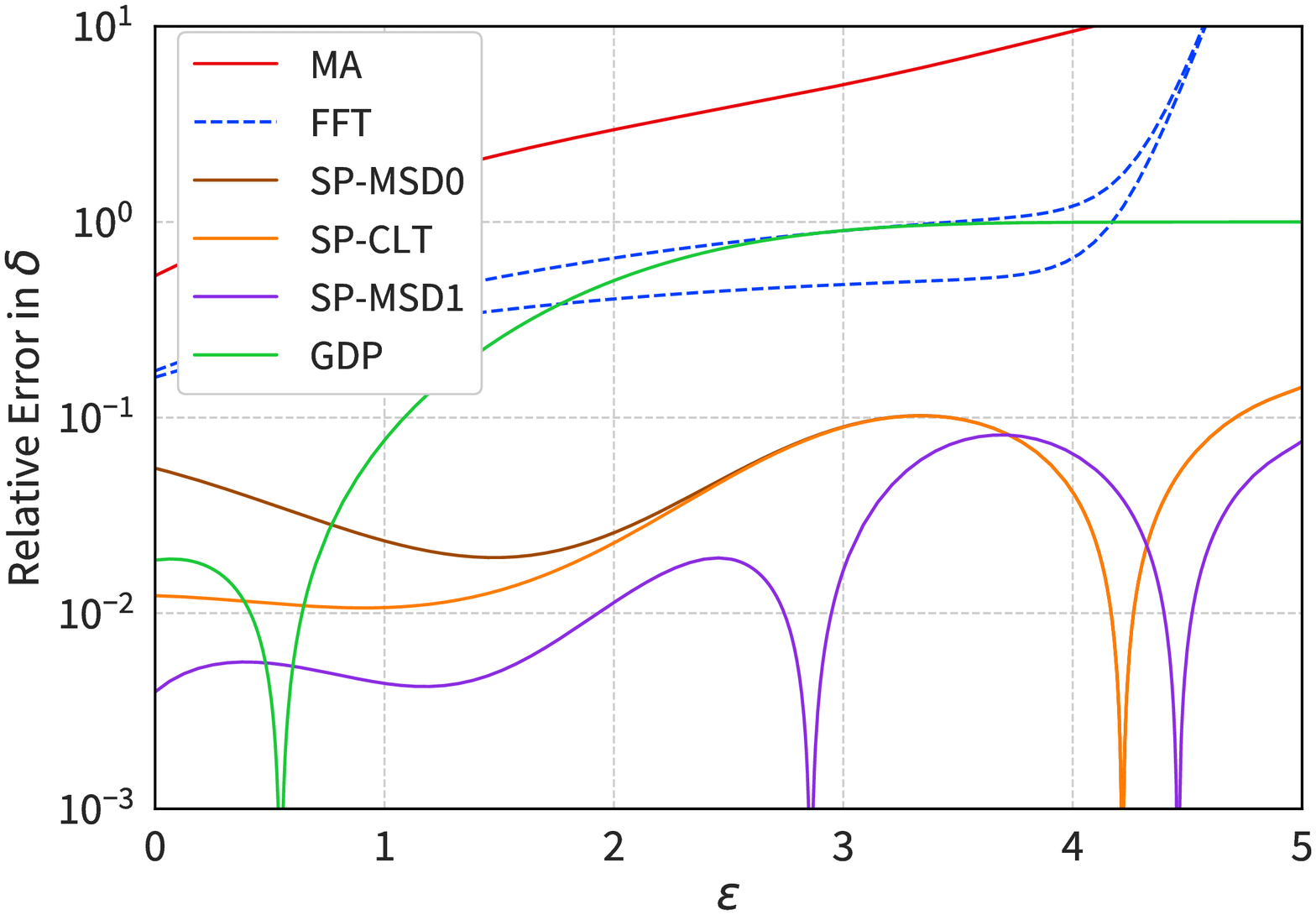}
        \captionsetup
        {width=0.9\textwidth,
        justification=centering}
        \caption*{(a) }
    \end{minipage}\hfill
    \begin{minipage}{0.5\textwidth}
        \centering
        \includegraphics[width=1\textwidth]{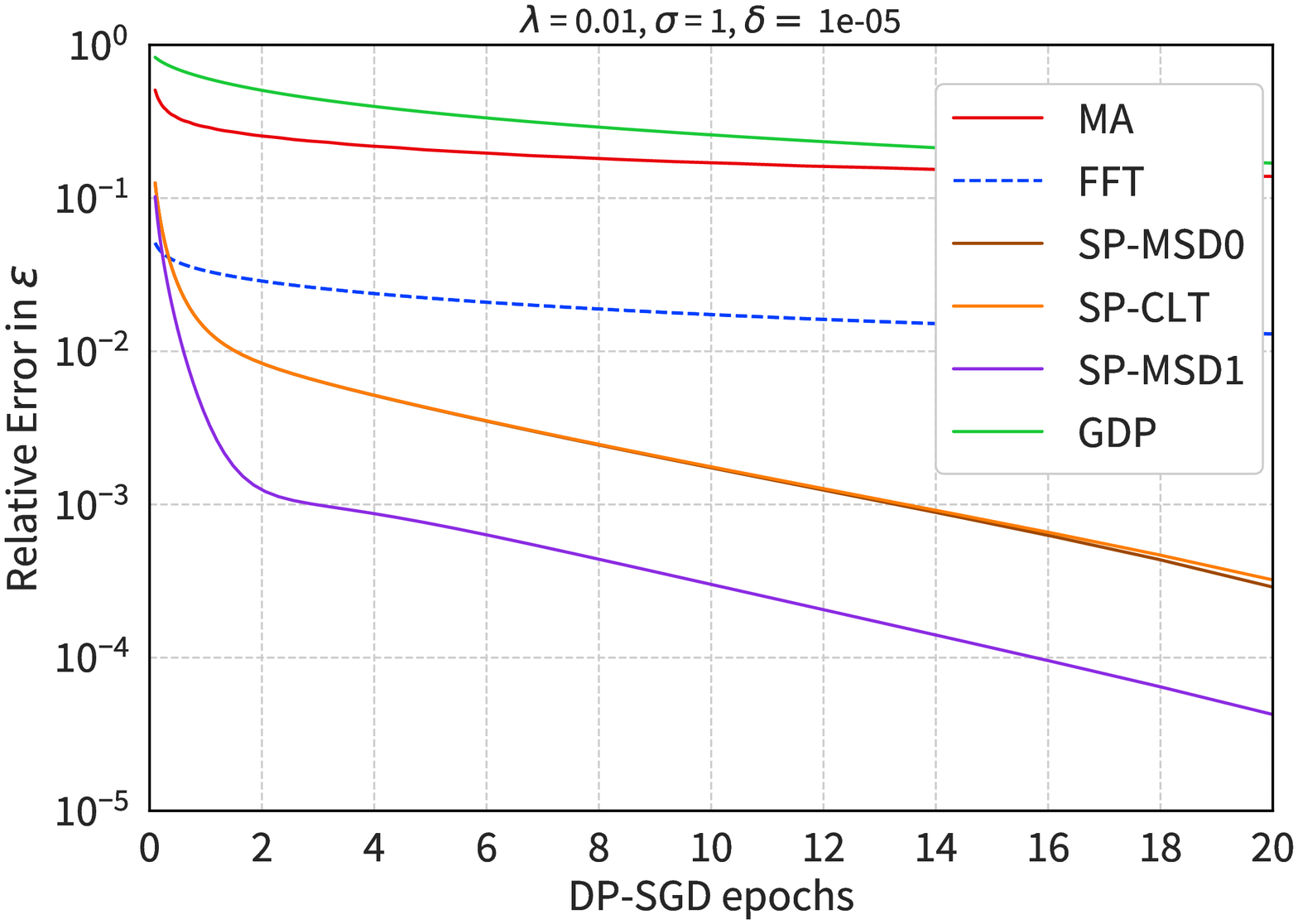} 
        \captionsetup{width=0.9\textwidth,
        justification=centering}
        \caption*{(b) }
    \end{minipage}
 
    \caption{The relative errors of the various accountants in Figure \ref{fig:DP-SGD-exp}. }
    \label{fig:dve}
\end{figure}

\paragraph{Experiment 1.} \label{exp1}
In Figure \ref{fig:DP-SGD-exp}(a), we approximate the privacy curve $\eps \mapsto \delta_{L^{(n)}}(\eps)$, fixing all other parameters. Specifically, we set $(\lambda,\sigma,n)= (10^{-2},1,2000)$. We plot three approximations to $\delta_{L^{(n)}}$ given by \eqref{SPA_approx1}, \eqref{SPA_approx2}, and \eqref{eq:true_saddlepoint_ncomp}, but label them all ``SPA [ours]'' since they are indistinguishable. Comparing against the baselines and ground truth, the {\sdp} approximations are more accurate than moments accountant and GDP while avoiding the floating point errors in the FFT accountant for $\delta \approx 10^{-10}$.

In Figure \ref{fig:DP-SGD-exp}(b), we account for how $\eps_{L^{(n)}}$ grows as a function of $n$, fixing all other parameters. Specifically, we fix a privacy budget parameter $\delta=10^{-5}$,  a subsampling rate $\lambda=10^{-2}$, and a noise standard-deviation $\sigma=0.65$. The number of compositions, $n$, is varied up to $2000$ (we start from $n=10$). We refer to the quantity $n\lambda$ as the number of \emph{epochs}, so it goes up to $20$. We again plot three approximations to $\delta_{L^{(n)}}$ given by \eqref{SPA_approx1}, \eqref{SPA_approx2}, and \eqref{eq:true_saddlepoint_ncomp}, and label them all ``SPA [ours].''

\paragraph{Experiment 2: Relative Error.} \label{exp2}
 We use the same data generated by Experiment~\hyperref[exp1]{1} to create relative-error plots. Specifically, we use our ``Truth'' calculation to illustrate the relative error of the three {\sdp} approximations to the true value. We do this for the baselines too, and the results are plotted in Figure~\ref{fig:dve}. For Figure $\ref{fig:dve}$(a), the relative error is given by $|1-\delta_{\text{acc}}/\delta_{\text{true}}|$, where $\delta_{\text{acc}}$ is from a given accountant, and $\delta_{\text{true}}$ is from ``Truth.'' Our best approximation $\dSDb{L^{(n)}}(\eps)$ (labeled ``SP-MSD1'' in the figure) approximates $\delta_{L^{(n)}}$ with a relative error that is orders of magnitude better than all the baselines for most values of $\epsilon$.\footnote{We note the dips in the relative error plot are due to the approximation oscillating between being below or above the composition curve, i.e., the relative error crosses 0 a few times.}

For Figure $\ref{fig:dve}$(b), the relative error is given by $|1-\eps_{\text{acc}}/\eps_{\text{true}}|$, where $\eps_{\text{acc}}$ is from a given accountant, and $\eps_{\text{true}}$ is from ``Truth.'' Both values of $\eps$ are computed using binary search. This figure shows that when DP-SGD has been run for more than 1 epoch, our SPA achieves a relative error of $<1\%$. Moreover, our most accurate approximation---which is obtainable from $\dSDb{L^{(n)}}(\eps)$ as given by~\eqref{eq:true_saddlepoint_ncomp} (and labeled ``SP-MSD1'' in the figures)---achieves relative errors $<0.1\%$ at 3 epochs and $<0.01\%$ at 16 epochs (i.e., it is $99.99\%$ accurate in a multiplicative sense after DP-SGD has been run for at least 16 epochs).

\begin{figure}[ht]
    \centering
    \begin{minipage}{0.5\textwidth}
        \centering
        \includegraphics[width=1\textwidth]{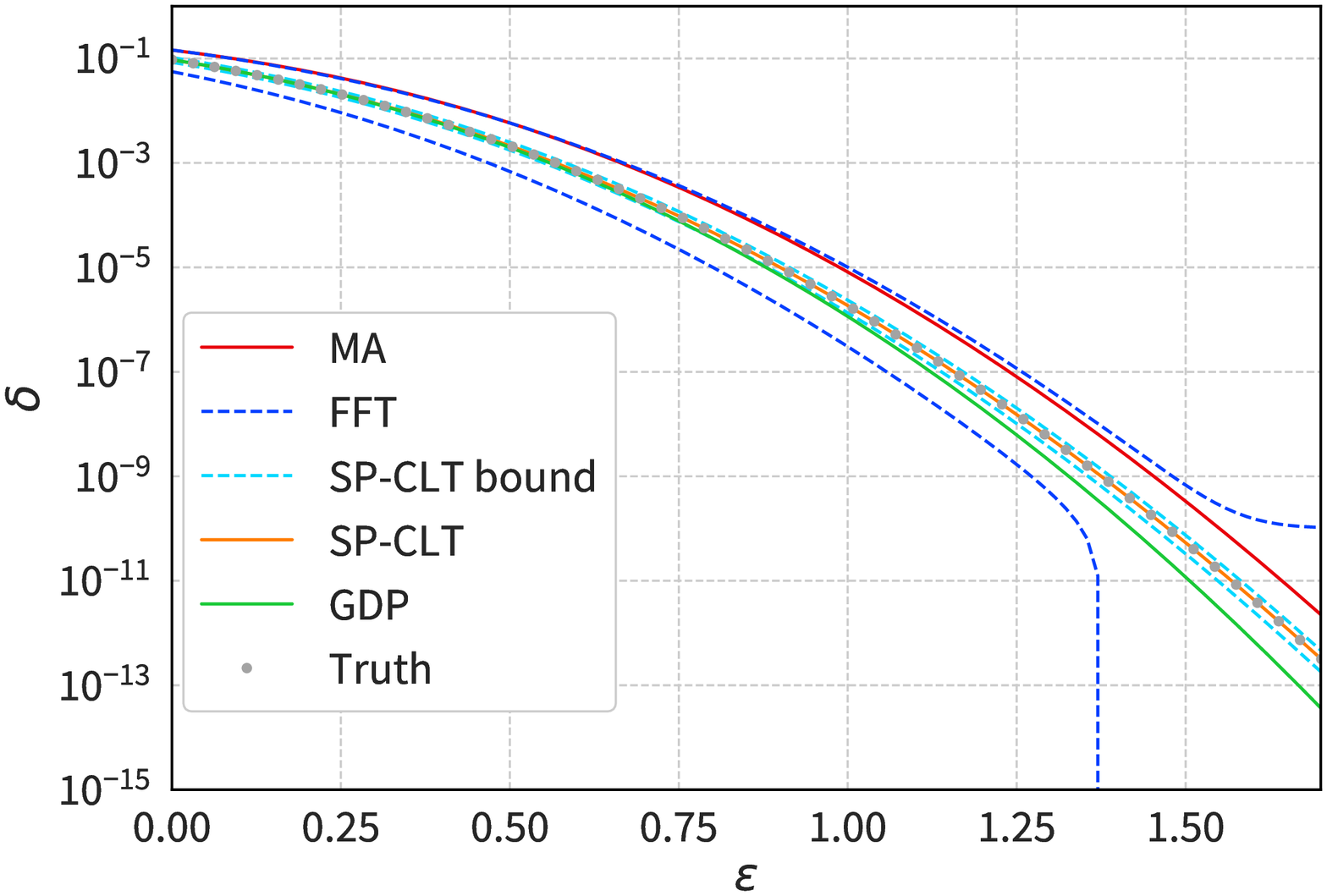} 
        \captionsetup{width=0.9\textwidth,
        justification=centering}
        \caption*{(a) }
    \end{minipage}\hfill
    \begin{minipage}{0.5\textwidth}
        \centering
        \includegraphics[width=1\textwidth]{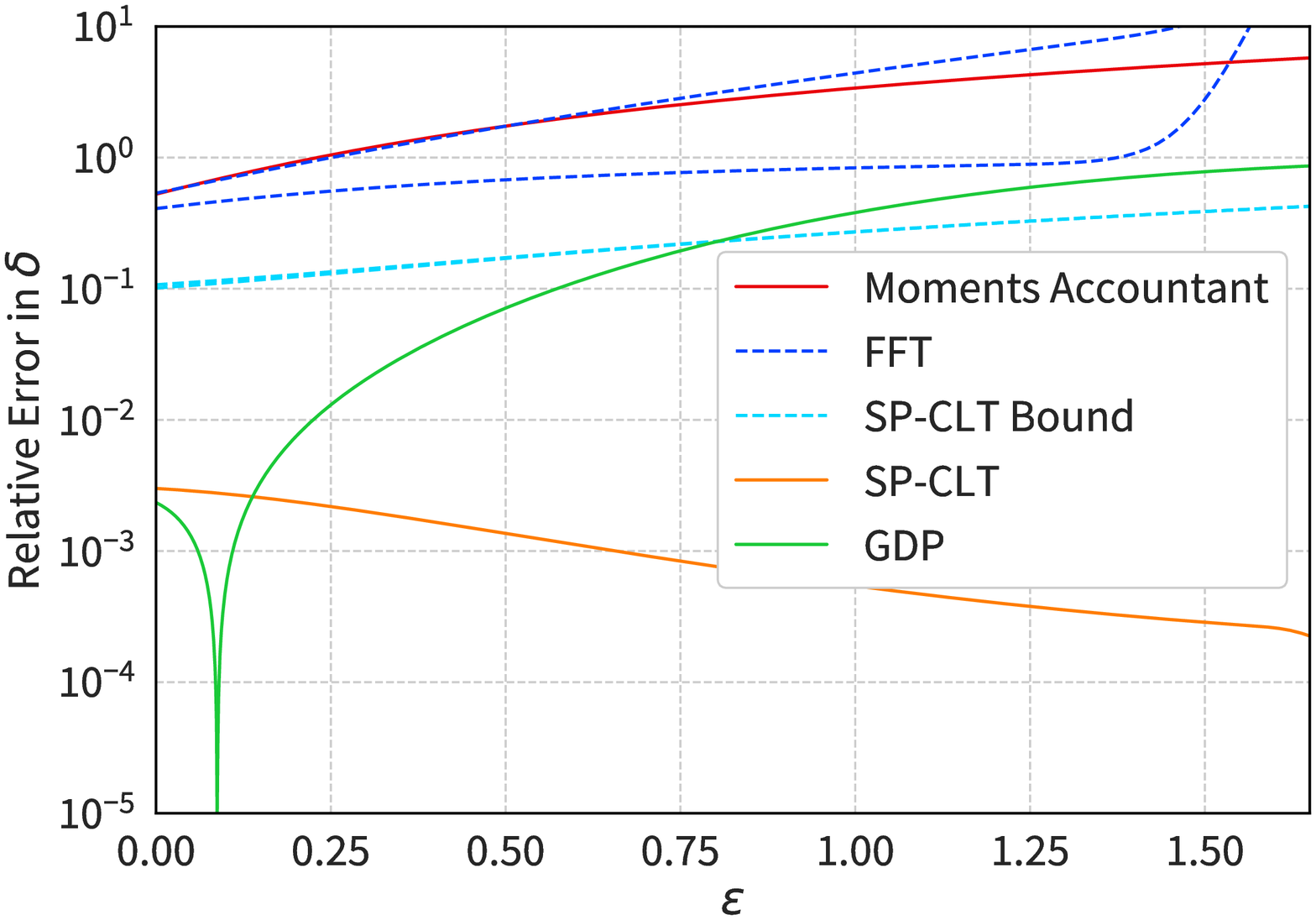} 
        \captionsetup{width=0.9\textwidth,
        justification=centering}
        \caption*{(b) }
    \end{minipage}
 
    \caption{The privacy curve of the subsampled Gaussian mechanism for $(\lambda,\sigma,n)= (10^{-2},2,2000)$. Pane (a) and (b)  highlights the tightness of our bounds on $\delta_{L^{(n)}, \,  \text{SP-CLT}}$ derived in Section \ref{sec:BE}.}
    \label{fig:SPA_bounds}
\end{figure}

\paragraph{Experiment 3: Bounds.} \label{exp3} We approximate the privacy curve $\eps \mapsto \delta_{L^{(n)}}(\eps)$ like in Figure \ref{fig:DP-SGD-exp}(a). Our goal this time is to illustrate the bounds derived for $\delta_{L^{(n)}, \,  \text{SP-CLT}}$ in Section \ref{sec:BE} and investigate how they compare to the baselines. Specifically, we set $(\lambda,\sigma,n)= (10^{-2},2,2000)$ and plot our upper and lower bounds on $\delta_{L^{(n)}, \, \text{SP-CLT}}(\eps)$ as a function of $\epsilon$. Our bounds are
$
    \delta_{L^{(n)}, \, \text{SP-CLT}}(\eps) \pm \text{err}_\text{SP}(\epsilon).
$
Although the error bounds are tight in the regime of Figure~\ref{fig:SPA_bounds}, there are parameter regimes (such as those in Figure  \ref{fig:dve}) where we have observed that the bounds underestimate the quality of the {\sdp} approximations. However, as shown in Figure~\ref{fig:dve}, the {\sdp} approximations are significantly more accurate than our error bounds suggest.

\section{Conclusion}\label{sec:conclusions}

We have introduced the saddle-point accountant (SPA) for DP. Via exponentially tilting the PLRV---and choosing the tilt in accordance with the saddle-point method from statistics---SPA combines the desirable behaviors of both large-deviation methods and central limit theorem based approaches. Two consequences follow: SPA outperforms both of the aforementioned methods, and it maintains the independence of the runtime in the setting of $n$-fold composition. We have demonstrated through numerical experiments that SPA shows comparable performance to state-of-the-art DP accountants.

\appendix

\section{Proof of Lemma~\ref{lem:subsampling}} \label{app:subsampling}

The case $\lambda=0$ is clear, so assume $\lambda \in (0,1]$. Suppose for now that $\gamma\cdot (1-\lambda)<1$. Denote $R \define T_sP$, and consider the function $G:(0,\infty)\to [0,\infty)$ defined by
\begin{equation}
    G(t) \define t \cdot \mathsf{E}_{1+\frac{\gamma-1}{t}}(P\| R).
\end{equation}
Since $\gamma'\mapsto \mathsf{E}_{\gamma'}(P\|R)$ is monotonically decreasing, we have that $G$ is monotonically increasing. Note that $0<\gamma \lambda + 1-\gamma \le \lambda$. Thus, plugging $t\in \{\gamma \lambda + 1-\gamma, \lambda\}$ into $G$, we obtain
\begin{equation} \label{eq:E}
    (\gamma \lambda + 1-\gamma)\cdot \mathsf{E}_{\frac{\gamma \lambda}{\gamma \lambda + 1 - \gamma}}(P \| R) \le \lambda \cdot \mathsf{E}_{\frac{\lambda-(1-\gamma)}{\lambda}}(P\|R).
\end{equation}
Now, note that
\begin{align}
    (\gamma \lambda + 1-\gamma)\cdot \mathsf{E}_{\frac{\gamma \lambda}{\gamma \lambda + 1 - \gamma}}(P \| R) &= (\gamma \lambda + 1-\gamma)\cdot \sup_{\mathcal{A}} P(\mathcal{A}) - \frac{\gamma \lambda}{\gamma \lambda + 1 - \gamma}\cdot R(\mathcal{A}) \\
    &= \sup_{\mathcal{A}} P(\mathcal{A}) -\gamma \cdot \left( (1-\lambda)P(\mathcal{A}) + \lambda R(\mathcal{A}) \right) \\
    &= \mathsf{E}_{\gamma}(P\|Q),
\end{align}
where the suprema are taken over all Borel sets $\calA \subset \BR^n$. In addition, by symmetry of $P$ around the origin, we have that
\begin{align}
    \mathsf{E}_{\gamma'}(P\|R) &= \sup_{\mathcal{A}} P(\mathcal{A}) - \gamma' P(\mathcal{A}-s) \\
    &= \sup_{\mathcal{A}} P(-\mathcal{A}) - \gamma' P(-\mathcal{A}-s) \\
    &= \sup_{\mathcal{A}} P(\mathcal{A}) - \gamma' P(\mathcal{A}+s) \\
    &= \sup_{\mathcal{A}} P(\mathcal{A}-s) - \gamma' P(\mathcal{A}) \\
    &= \mathsf{E}_{\gamma'}(R\|P).
\end{align}
Therefore,
\begin{align}
    \lambda \cdot \mathsf{E}_{\frac{\lambda-(1-\gamma)}{\lambda}}(P\|R) &= \lambda \cdot \mathsf{E}_{\frac{\lambda-(1-\gamma)}{\lambda}}(R\|P) \\
    &=\lambda \cdot \sup_{\mathcal{A}} R(\mathcal{A}) - \frac{\lambda-(1-\gamma)}{\lambda} \cdot P(\mathcal{A}) \\
    &= \sup_{\mathcal{A}} \ ((1-\lambda)P(\mathcal{A}) + \lambda R(\mathcal{A})) - \gamma P(\mathcal{A}) \\
    &= \mathsf{E}_\gamma(Q\|P).
\end{align}
We conclude from~\eqref{eq:E} the desired inequality $\mathsf{E}_\gamma(P\|Q) \le \mathsf{E}_{\gamma}(Q\|P)$. In addition, the case $\gamma \cdot (1-\lambda) \ge 1$ follows immediately since then $\mathsf{E}_\gamma(P\|Q) =0\le \mathsf{E}_{\gamma}(Q\|P)$.

\section{Proof of Theorem~\ref{thm:BE}} \label{app:BE}

Fix a tilting parameter $t>0$ such that $\bbE[e^{tL_j}]<\infty$ for all $j$. Recall from \eqref{eq:MAvar} that
\begin{equation}
       \delta_{L^{(n)}}(\epsilon) = e^{K_{L^{(n)}}(t)-\epsilon t}~ \bbE\left[\bar{f}\left(\tilL^{(n)} -\epsilon ,t\right) \right]
\end{equation}
where $\tilL^{(n)}$ is the exponential tilting of $L^{(n)}$ with parameter $t$, and
\begin{equation}
    \bar{f}(x,t)=e^{-xt}(1-e^{-x})^+
\end{equation}
Note that $\tilL^{(n)}=\tilL_1+\cdots+\tilL_n$. Moreover $K'_{L^{(n)}}(t)=\bbE[\tilL^{(n)}]$ and $K''_{L^{(n)}}(t)=\var[\tilL^{(n)}]$. We consider the function $\barf(x,t)$. Fig.~\ref{fig:crudeapprox} illustrates that, for fixed $t$, $\barf(x,t)$ is a unimodal function with a maximal value of $t^t/(t+1)^{t+1}$. This fact was used in the proof of \eqref{eq:min_1}; here we prove it formally. Certainly $\barf(x,t)\ge 0$ for all $x$. For $x>0$ the derivative (with respect to $x$) is
\be
\barf'(x,t)=-t e^{-tx}(1-e^{-x})+e^{-tx}e^{-x} =e^{-tx}\left[-t+(t+1)e^{-x}\right].
\ee
Note that $-t+(t+1)e^{-x}$ is monotonically decreasing in $x$, which means that $\barf(x,t)$ is increasing until $-t+(t+1)e^{-x}=0$, and is subsequently decreasing. In particular, the maximal value of $\barf$ is attained when
\be
x=x_0=-\log \frac{t}{t+1}.
\ee
Note that $x_0>0$. Thus, the maximal value of $\barf$ is
\begin{align}
f_{\max}\define \barf(x_0,t)=
\barf\left(-\log \frac{t}{t+1},t\right)
&=\left( \frac{t}{t+1}\right)^t \left(1-\frac{t}{t+1}\right) = \frac{t^t}{(t+1)^{t+1}}.
\label{eq:fmax}
\end{align}
Thus, between $x=0$ and $x=x_0$, $\barf(x,t)$ is monotonically increasing from $0$ to $f_{\max}$; then from $x=x_0$ to $x=\infty$, $\barf(x,t)$ is monotonically decreasing from $f_{\max}$ to $0$. Thus, there exist functions $f_1^{-1}(z)$, $f_2^{-1}(z)$ such that, for any $z\in(0,f_{\max})$, $\barf(x,t)>z$ if and only if
\[
f_1^{-1}(z)<x<f_2^{-1}(z).
\]
Therefore,
\be
\bbE[ \barf(\tilL^{(n)}-\eps,t)]
=\int_0^{f_{\max}} \bbP\left[\barf(\tilL^{(n)}-\veps,t)>z\right]dz
=\int_0^{f_{\max}} \bbP\left[f_1^{-1}(z)<\tilL^{(n)}-\veps<f_2^{-1}(z)\right]dz.
\ee
Recall that $\tilL^{(n)}=\sum_{j=1}^n \tilL_j$ where the $\tilL_j$ are independent, and note that $\bbE[\tilL_j]=K_{L_j}'(t)$ and $\var[\tilL_j]=K_{L_j}''(t)$. Thus we may apply the Berry-Esseen theorem to write
\be
\sup_{x\in \BR} \left| \bbP\left[\tilL^{(n)}>x\right] -  \bbP[Z>x] \right| \le \frac{0.56 \ \uP_t^{(n)}}{K_{L^{(n)}}''(t)^{3/2}}
\ee
where $Z\sim\calN\left( K_{L^{(n)}}'(t),K_{L^{(n)}}''(t)\right)$ and $\uP_t^{(n)}$ is defined in \eqref{eq:Pt_def}. Thus we have the upper bound
\begin{align}
\delta_{L^{(n)}}(\veps) &= e^{K_{L^{(n)}}(t)-\veps t} \  \bbE \left[\barf\left(\tilL^{(n)}-\veps,t\right)\right] \\
&= e^{K_{L^{(n)}}(t)-\veps t} \int_0^{f_{\max}} \bbP\left[f_1^{-1}(z)<\tilL^{(n)}-\veps<f_2^{-1}(z)\right]dz
\\
&\le e^{K_{L^{(n)}}(t)-\veps t} \ \left(\int_0^{f_{\max}} \BP\left[f_1^{-1}(z)<Z-\veps<f_2^{-1}(z)\right]dz+\frac{1.12f_{\max}\uP_t^{(n)}}{K_{L^{(n)}}''(t)^{3/2}}\right)
\\
&=e^{K_{L^{(n)}}(t)-\veps t} \ \left(\bbE \left[\barf\left(Z-\veps,t\right)\right]+\frac{1.12f_{\max}\uP_t^{(n)}}{K_{L^{(n)}}''(t)^{3/2}}\right) \label{eq:BE upper}
\end{align}
Similarly, we have the lower bound
\begin{align}
\delta_{L^{(n)}}(\veps) \ge e^{K_{L^{(n)}}(t)-\veps t} \ \left(\bbE \left[\barf\left(Z-\veps,t\right)\right]-\frac{1.12f_{\max}\uP_t^{(n)}}{K_{L^{(n)}}''(t)^{3/2}}\right). \label{eq:BE lower}
\end{align}

\section{Proof of Proposition~\ref{prop:gaussian_computation}}\label{app:gaussian_computation}

Denote $K=K_{L^{(n)}}$ for short. The Gaussian expectation may be computed as
\begin{align}
    &\bbE \left[\barf\left(Z-\veps,t\right)\right] = \exp\left(\frac{K''(t)t^2}{2}-(K'(t)-\eps)t\right) Q\left(\sqrt{K''(t)}\,t-\frac{K'(t)-\eps}{\sqrt{K''(t)}}\right)\nonumber \\
    &\qquad -\exp\left(\frac{K''(t)(t+1)^2}{2}-(K'(t)-\eps)(t+1)\right) Q\left(\sqrt{K''(t)}\,(t+1)-\frac{K'(t)-\eps}{\sqrt{K''(t)}}\right).
\end{align}
Using $Q(z) = \frac{q(z)}{\sqrt{2\pi}} e^{-z^2/2}$ and the definitions of $\alpha,\beta,\gamma$, we get
\begin{equation}
    \bbE \left[\barf\left(Z-\veps,t\right)\right] = \frac{q(\alpha) - q(\beta)}{\sqrt{2\pi}} \ e^{-\gamma^2/2}.
\end{equation}
Plugging this into the definition of $\delta_{L^{(n)},\text{SP-CLT}}$ completes the proof.

\section{Proof of Theorem~\ref{thm:err}} \label{app:err}

We write $K=K_{L^{(n)}}$, $L=L^{(n)}$, and $\uP_t = \uP_t^{(n)}$, for short. Recall the definition of the error term in~\eqref{eq:berry_esseen_delta_error}
\begin{align}
    \text{err}_{\text{SP}}(\eps)&
    =
    \exp\left(K(t_0)-\eps t_0\right)\frac{t_0^{t_0}}{(1+t_0)^{1+t_0}}\cdot \frac{1.12\, \uP_{t_0}}{K''(t_0)^{3/2}}.
\end{align}
From the characterization of the saddle point in Proposition~\ref{lem:t}, we have that
\begin{equation}
    t_0 \sim \frac{b+\sqrt{b^2+4}}{2\sigma_L}.
\end{equation}
By Assumption~\ref{assumption:KL V P}, we have that $\sigma_L^2 = K''(0) \sim n\V$ as $n\to \infty$. Hence, $t_0 \sim c/\sqrt{n}$ for $c=(b+\sqrt{b^2+4})/(2\V) = o(n^{1/6})$. Thus, by Assumption~\ref{assumption:KL V P} again, $(K''(t_0),\uP_{t_0})\sim n \cdot ( \V,\uP)$. As we also have that $t_0 \to 0$, we conclude that 
\begin{equation}
    \frac{t_0^{t_0}}{(1+t_0)^{1+t_0}}\cdot \frac{1.12\, \uP_{t_0}}{K''(t_0)^{3/2}} \sim \frac{1.12 \ \uP}{\V^{3/2} \cdot \sqrt{n}}.
\end{equation}
Thus, it only remains to analyze the asymptotic of $\exp\left( K(t_0)-\veps t_0 \right)$.

We use the following Taylor expansion of $K$ around $0$:
\begin{equation}
    K(t_0) = t_0 \cdot \BE[L] + \frac{t_0^2}{2} \cdot \sigma_L^2 + \frac{t_0^3}{6} \cdot K'''(\xi),
\end{equation}
where $0\le \xi \le t_0$. Using $\veps = \BE[L] + b \sigma_L$, and writing $t_0 = d_0/\sigma_L$ (so $d_0 \sim  (b+\sqrt{b^2+4})/2$ by Proposition~\ref{lem:t}), we obtain
\begin{equation}
    K(t_0) - \veps t_0 = \frac{d_0^2}{2} - b d_0 + \frac{d_0^3K'''(\xi)}{6\sigma_L^3}.
\end{equation}

Now, note that $K'''(\xi) = \sum_{j=1}^n K_j'''(\xi)$. Thus, applying the triangle inequality, we obtain that $|K'''(\xi)|\le \uP_{\xi}$. As $0\le \xi \le t_0$, Assumption~\ref{assumption:KL V P} yields that $|K'''(\xi)| = O(n)$. As $\sigma_L = \Theta(\sqrt{n})$, and $d_0 = o(n^{1/6})$, we infer that
\begin{equation}
    \frac{d_0^3K'''(\xi)}{6\sigma_L^3} \to 0
\end{equation}
as $n\to \infty$. Hence, 
\begin{equation}
    \exp\left(K(t_0)-\eps t_0\right) \sim \exp\left( \frac{d_0^2}{2} - bd_0 \right).
\end{equation}
Writing $d_0 = \tau_0 \cdot (b+\sqrt{b^2+4})/2$, so $\tau_0>0$ and $\tau_0 \to 1$ by Proposition~\ref{lem:t}, then collecting terms, we obtain
\begin{equation}
    \frac{d_0^2}{2} - bd_0 = \frac{\tau_0^2}{2} -  (2-\tau_0)\tau_0 \cdot \frac{b^2+b\sqrt{b^4+4}}{4}.
\end{equation}
Therefore, we obtain that
\begin{equation}
     \exp\left(K(t_0)-\eps t_0\right) \sim \frac{\sqrt{e}}{C(b)^{\tau}}
\end{equation}
where $\tau := (2-\tau_0)\tau_0 \to 1$. Putting the asymptotics shown above together, we conclude that
\begin{equation}
    \errt(\veps) \sim \frac{1.12\sqrt{e} \ \uP }{\V^{3/2} \cdot C(b)^\tau \cdot \sqrt{n}},
\end{equation}
as desired.

\section{Proof of Proposition~\ref{lem:t}: Asymptotic of the Saddle Point} \label{app:t}

We write $K=K_{L^{(n)}}$ and $L=L^{(n)}$ for short. Consider the saddle-point equation~\eqref{eq:tau0}:
\begin{equation}
    K'(t) = \veps + \frac{1}{t} + \frac{1}{1+t}.
\end{equation}
The left-hand side strictly increases from $\BE[L]$ to $\esup L$ over $t\in [0,\infty)$, whereas the right-hand side strictly decreases from $\infty$ to $\veps$ over the same interval. Hence, there exists a unique solution $t=t_0>0$, which we call the saddle point.

We show first that $t_0 \to 0$ as $n\to \infty$. Suppose, for the sake of contradiction, that $t^* \define \limsup_{n\to \infty} t_0 > 0$, and let $n_k\nearrow \infty$ be a sequence of indices such that the sequence of the $n_k$-th saddle points, denoted $t_0^{(k)}$, converge to $t^*$. Let $\rho_2:(0,\infty) \to (0,\infty)$ be defined by $\rho_2(t) \define (K'(t)-\BE[L])/(t\sigma_L^2)$, so $\rho_2(t) \to 1$ as $t\to 0^+$ and
\begin{equation}
    K'(t) = \BE[L] + \sigma_L^2 t \rho_2(t).
\end{equation}
Note that $\rho_2$ is a continuous function. Noting that $\veps = \BE[L] + b\sigma_L$, rearranging the saddle-point equation yields that
\begin{equation} \label{ra}
    \frac{1+ \frac{\sigma_L^2}{\BE[L]} t\rho_2(t)}{1+ b \frac{\sigma_L}{\BE[L]}} = 1 + \frac{1}{\veps t} + \frac{1}{\veps \cdot (1+t)}.
\end{equation}
Taking $t\in \{t_0^{(k)}\}_{k\in \BN}$, letting $k\to \infty$, and recalling the assumptions that $(\BE[L],\sigma_L^2)\sim n\cdot (\KL,\V)$ for $\KL,\V>0$ and that $b=o(\sqrt{n})$, we infer from~\eqref{ra} that
\begin{equation} \label{rb}
    \frac{\V t^* \rho_2(t^*)}{\KL} = 0.
\end{equation}
Equality~\eqref{rb} contradicts that $\V,t^*,\rho_2(t^*),\KL>0$. Thus, we must have that $t^*=0$.

Consider the reparametrization $t=d/\sigma_L$, so $d$ is a variable over $(0,\infty)$. The saddle-point equation can be rewritten as
\begin{equation} \label{eq:SP quadratic}
    \left( \rho_2(t) - \frac{b}{\sigma_L} \right) d^2 - \left( b + \frac{2}{\sigma_L} \right) d - \left( 1- \frac{\rho_2(t)d^3}{\sigma_L} \right) = 0.
\end{equation}
We rewrite the saddle-point equation in this ``quadratic'' form since it closely approximates the quadratic $d^2-bd-1=0$ at the saddle-point. Indeed, let $d_0>0$ be such that $t_0=d_0/\sigma_L$. We obtain from~\eqref{eq:SP quadratic} the inequality $\frac12 d_0^2 - (b+1)d_0 -1 \le 0$ for all large $n$. This latter inequality yields that
\begin{equation}
    d_0 \le b+1 + \sqrt{(b+1)^2+2} = o(n^{1/6}).
\end{equation}
Hence, $\rho_2(t_0)d_0^3/\sigma_L \to 0$ as $n\to \infty$, i.e., the ``constant'' term in~\eqref{eq:SP quadratic} approaches $1$. Thus, for all large $n$, completing the square in~\eqref{eq:SP quadratic} yields
\begin{equation} \label{eq:d}
    d_0 = \frac{b + \frac{2}{\sigma_L} + \sqrt{\left( b+ \frac{2}{\sigma_L} \right)^2 + 4 \left( 1 - \frac{\rho_2(t_0)d_0^3}{\sigma_L} \right) \left( \rho_2(t_0) - \frac{b}{\sigma_L} \right)}}{2\left( \rho_2(t_0) - \frac{b}{\sigma_L} \right)}.
\end{equation}
Taking $n\to \infty$, we obtain
\begin{equation}
    d_0 \sim  \frac{b+\sqrt{b^2+4}}{2},
\end{equation}
which gives the desired asymptotic formula for the saddle-point $t_0=d_0/\sigma_L$.

\bibliography{main}
\bibliographystyle{alpha}

\end{document}